\documentclass[runningheads]{llncs}

\usepackage{amsmath}%

\usepackage{amsthm}
\usepackage{amsfonts}%
\usepackage{amssymb}%
\usepackage{tikz}
\usepackage{array,multirow}
\usepackage{tikz}

\usetikzlibrary{calc,trees,positioning,arrows,chains,shapes.geometric,%
decorations.pathreplacing,decorations.pathmorphing,shapes,%
matrix,shapes.symbols,plotmarks,decorations.markings,shadows}
\newcommand{\set}[1]{\left\{ #1 \right\}}
\newcommand{\card}[1]{\left| #1 \right|}

\newcommand{\mcald}{\mathcal{D}}

\newcommand{\mcalq}{\mathcal{Q}}

\newcommand{\mcalz}{\mathcal{Z}}
\newcommand{\mcali}{\mathcal{I}}
\newcommand{\mcalj}{\mathcal{J}}

\newcommand{\stcut}[1][S,T]{$(#1)$-cut }

\newcommand{\cmincuts}{\textsc{counting mincuts}}

\newcommand{\closest}{closest }
\newcommand{\anodyne}{\closest }
\newcommand{\network}{drainage }

\begin{document}

\title{Fixed-parameter tractability of counting small minimum $(S,T)$-cuts}
\titlerunning{Fixed-parameter tractability of counting small minimum $(S,T)$-cuts}

\author{Pierre Berg\' e\inst{1}, Benjamin Mouscadet\inst{2}, Arpad Rimmel\inst{2} and Joanna Tomasik\inst{2}}
\institute{LRI, Universit\'e Paris-Sud, Universit\'e Paris-Saclay, Orsay France \email{Pierre.Berge@lri.fr}
\and LRI, CentraleSup\' elec, Universit\'e Paris-Saclay, Orsay France,
\email{Arpad.Rimmel@lri.fr},
\email{Benjamin.Mouscadet@supelec.fr},
\email{Joanna.Tomasik@lri.fr}}

\maketitle

\begin{abstract}
The parameterized complexity of counting minimum cuts stands as a natural question because Ball and Provan showed its \#P-completeness.
For any undirected graph $G=(V,E)$ and two disjoint sets of its vertices $S,T$, we design a fixed-parameter tractable algorithm which counts minimum edge $(S,T)$-cuts parameterized by their size $p$.

Our algorithm operates on a transformed graph instance. This transformation, called drainage, reveals a collection of at most $n=\card{V}$ successive minimum $(S,T)$-cuts $Z_i$. We prove that any minimum $(S,T)$-cut $X$ contains edges of at least one cut $Z_i$. This observation, together with Menger's theorem, allows us to build the algorithm counting all minimum $(S,T)$-cuts with running time $2^{O(p^2)}n^{O(1)}$. 
Initially dedicated to counting minimum cuts, it can be modified to obtain an FPT sampling of minimum edge $(S,T)$-cuts.
\end{abstract}

\section{Introduction}

The issue of counting minimum cuts in graphs has been drawing attention over the years due to its practical applications. Indeed, the number of minimum cuts is an important factor for the network reliability analysis~\cite{BaCoPr95,BaPr83,BaPr84,NaSuIb91}. Thereby, the probability that a stochastic graph is connected may be computed~\cite{BaPr83}. Furthermore, cuts on planar graphs are used for image segmentation~\cite{BoVe06}. An image is seen as a planar graph where vertices represent pixels and edges connect two neighboring pixels if they are similar. Counting minimum cuts provides an estimation of the number of segmentations.

We focus on the problem of counting minimum edge $(S,T)$-cuts in undirected graphs $G = (V,E)$, $S,T \subseteq V$. We call it \cmincuts\ (Def.~\ref{def:cmincuts}) as it is the counting variant of the classical problem \textsc{mincut}, which asks for a minimum $(S,T)$-cut in graph $G$. 
Ball and Provan showed in~\cite{PrBa83} that \cmincuts\ is unlikely solvable in polynomial time as it is \#P-complete. 
They also devised a polynomial-time algorithm for \cmincuts\ on planar graphs~\cite{BaPr83}. Bez\'{a}kov\'{a} and Friedlander~\cite{BeFr12} generalized it with an $O(n\mu+n\log n)$-time algorithm on weighted planar graphs, where $\mu$ is the length of the shortest $(s,t)$-paths. For general graphs, some upper bounds on the number of minimum cuts have been given~\cite{ChRa02} in function of parameters such as the radius, the maximum degree, etc. Two fixed-parameter tractable (FPT) algorithms have been proposed for \cmincuts . Bez\'{a}kov\'{a} {\em et al.}~\cite{BeChFo16} built an algorithm for both directed and undirected graphs with small treewidth $\lambda$; its time complexity is $O(2^{3\lambda}\lambda n)$. Moreover, Chambers {\em et al.}~\cite{ChFoNa13} designed an algorithm for directed graphs embedded on orientable surfaces of genus $g$: its execution time is $O(2^gn^2)$. We study the fixed-parameter tractability of \cmincuts , parameterized by the size $p$ of the minimum $(S,T)$-cuts.

\begin{definition}[Counting mincuts]
~

\textbf{Input: }Undirected graph $G=(V,E)$, sets of vertices $S,T \subsetneq V$, $S\cap T = \emptyset$.  

\textbf{Output: }The number of minimum edge $(S,T)$-cuts. 
\label{def:cmincuts}
\end{definition}

The minimum $(S,T)$-cut size for a \cmincuts\ instance $\mcali = (G,S,T)$ is obtained in polynomial time~\cite{FoFu56}. A brute force XP algorithm computes the number $C(\mcali)$ of minimum $(S,T)$-cuts in time $n^{O(p)}$ by enumerating all edge sets of size $p$ and picking up those which are $(S,T)$-cuts. More efficient exponential algorithms exist, as the one of Nagamochi {\em et al.}, in time $O\left(pn^2 + pnC(\mcali)\right)$, in~\cite{NaSuIb91}. Our contribution, summarized in the theorem below, is an algorithm efficient for small values of $p$.
\begin{theorem}
The counting of minimum edge $(S,T)$-cuts can be solved in time $O(2^{p(p+2)}pmn^3)$ on undirected graphs $G=(V,E)$, where $n = \card{V}$ and $m = \card{E}$.
\label{th:main_result}
\end{theorem}

An FPT$\langle p \rangle$ algorithm can be deduced from the results in two articles~\cite{BeChFo16,MaORa13} and its execution time is $O^*\left(2^{H(p)}\right)$ where $H(p) = \Omega\left(\frac{2^p}{\sqrt{p}}\right)$. The treewidth reduction theorem established by Marx {\em et al.} in~\cite{MaORa13} says that there is a linear-time reduction transforming graph $G$ into another graph $G'$ which conserves the $(s,t)$-cuts of size $p$ and such that the treewidth of $G'$, $\tau(G')$, verifies $\tau(G') = 2^{O(p)}$. After this transformation, the number of minimum $(S,T)$-cuts of $G'$ is obtained thanks to the algorithm given in~\cite{BeChFo16}. The overall time taken with this method is $O^*\left(2^{2^p}\right)$. Our result, Theorem~\ref{th:main_result}, improves this exponential factor.

This result highlights a complexity gap between the counting and the enumeration, as the latter cannot be FPT parameterized by $p$. 
Indeed, certain instances contain a number $C(\mcali) = (\frac{n-1}{p})^p$ of minimum cuts, as in case of graph $G$ made of $p$ vertex-disjoint $(S,T)$-paths with $S =\set{s}$ and $T =\set{t}$.

Our algorithm is based on a cut-decomposition $\mcalz(\mcali) = \left(Z_1,\ldots,Z_k\right)$ of instance $\mcali$, $1 \le k < n$, called the \textit{drainage} where for every $1\le i\le k$, edge set $Z_i$ is a minimum $(S,T)$-cut. Set $R(Z_i,S)$ denotes the vertices which are reachable from $S$ after the removal of edges in $Z_i$. The reachable sets of $Z_i$ are included one into another: $R(Z_1,S) \subsetneq R(Z_2,S) \subsetneq \ldots \subsetneq R(Z_k,S)$. The \network fulfils the following property: if $X$ is a minimum $(S,T)$-cut, some edges $B_i$ of a certain $Z_i$ belong to $X$, $B_i = X \cap Z_i \neq \emptyset$, and no other edge of $X$ has one endpoint in $R(Z_i,S)$. The set $B_i$ is called the \textit{front dam} of cut $X$. The key idea of the recursive counting we propose is that any minimum cut $X$ is the union of its front dam with a minimum cut of a sub-instance, called \textit{dry instance}, of $\mcali$. These techniques work as well on multigraphs, {\em i.e.} on graphs with multiple edges.
After modifications, our algorithm also samples minimum edge $(S,T)$-cuts.

To design the \network $\mcalz(\mcali)$, we use the concept of \textit{important cuts}~\cite{Ma06} which is the key ingredient of many FPT algorithms to solve cuts problems~\cite{BoDaTh11,ChHaMa13,CyLoPiPiSa14,Ma06,MaRa11}. An $(S,T)$-cut $Y$ is \textit{important} if there is no other $(S,T)$-cut $Y'$ such that $\card{Y'} \le \card{Y}$ and $R(Y,S) \subsetneq R(Y',S)$.
There is a unique minimum important $(S,T)$-cut and it can be identified in polynomial time~\cite{Ma06}.

The second concept used in our algorithm is Menger's theorem~\cite{Me27}. It states that the size of minimum edge $(S,T)$-cuts in an undirected graph is equal to the maximum number of edge-disjoint $(S,T)$-paths. 
As the max-flow min-cut theorem~\cite{FoFu56} generalizes Menger's theorem, one of the largest sets of edge-disjoint $(S,T)$-paths is found in polynomial time.

To close this introductory chapter, we give a ``table of contents'' of our article. Section~\ref{sec:preliminaries} introduces the notations used. Section~\ref{sec:skeleton} explains the construction of the \network $\mcalz(\mcali) = (Z_1,\ldots,Z_k)$. In Section~\ref{sec:algo}, we propose our algorithm and compute its time complexity. 
Finally, we conclude and give ideas about future research.

\section{Definitions and notation} \label{sec:preliminaries}

We summarize basic concepts of parameterized and counting complexity but also introduce the notation we will use.

\textbf{Fixed-parameter tractability.} NP-hard problems are unlikely to be solvable with polynomial time algorithms. However, solving them efficiently may become possible when parameters are associated to problem instances and the values of these parameters are small.

Referring to Downey and Fellows~\cite{DoFe99} and Niedermeier~\cite{Ni06}, a parameterized problem is said \textit{fixed-parameter tractable} (FPT) if there is an algorithm solving it in time $O(f(p)P(n)) = O^*(f(p))$, where $p$ is a parameter, $n$ is the instance size, $P$ is a polynomial function, and $f$ is an arbitrary computable function. As a problem may be studied for different parameters $p_1,p_2,\ldots$, the notation ``FPT'' becomes ambiguous. If there is an algorithm solving a problem in time $O\left(f(p_1)P(n)\right)$, then it is FPT$\langle p_1\rangle$. In this study, the parameter $p$ of \cmincuts\ is the size of the minimum $(S,T)$-cut.

\textbf{Counting problems.} The study of \#P complexity class and the counting problems it contains, started with Valiant~\cite{Va79}. Class \#P is the set of counting problems such that their decision version is in class NP. The subclass \#P-complete contains counting problems such that all problems in \#P can be reduced to them with a polynomial-time counting reduction. No \#P-complete problem can be solved in polynomial time unless P$=$NP. Moreover, there are decision problems such as \textsc{cnf-2sat}~\cite{Kr67} which are solvable in polynomial time but their associated counting problem is \#P-complete \cite{Va79}.
The complexity of counting problems has been extended via the parameterized complexity framework~\cite{Cu18,FlGr04}. A relevant question to ask about a \#P-complete problem is whether there is an FPT algorithm counting all its solutions. For example, with $G$ and $H$ as input, FPT algorithms counting the number of occurrences of $H$ as a subgraph of $G$ have been intensively studied~\cite{ArRa02,GuSi13,WiWi13}.

\textbf{Cuts in undirected graphs.} We study undirected graphs $G=(V,E)$, where $n = \card{V}$ and $m = \card{E}$. For any set of vertices $U \subseteq V$, we denote by $E\left[U\right]$ the set of edges of $G$ with two endpoints in $U$ and $G\left[U\right]$ the subgraph of $G$ induced by $U$: $G\left[U\right] = \left(U,E\left[U\right]\right)$. Notation $G\backslash U$ refers to the graph deprived of vertices in $U$. For any set of edges $E' \subseteq E$, the graph $G$ deprived of edges $E'$ is denoted by $G\backslash E'$:
\[
G\backslash U = G\left[V\backslash U\right]~~\mbox{and}~~G\backslash E' = \left(V,E\backslash E'\right).
\]
A \textit{path} is a sequence of pairwise different vertices $v_1\cdot v_2\cdot v_3\cdots v_i\cdot v_{i+1} \cdots$, where two successive vertices $(v_i,v_{i+1})$ are adjacent in $G$. To improve readability, we abuse notations: $v_1 \in Q$ and $(v_1,v_2) \in Q$ mean that vertex $v_1$ and edge $(v_1,v_2)$ are on path $Q$, respectively.

Cut problems usually consist in finding the smallest set of edges $X \subseteq E$ which splits the graph $G\backslash X$ into connected components meeting certain criteria. Given two sets of vertices $S$ (sources) and $T$ (targets), set $X \subseteq E$ is an \stcut if there is no path connecting a vertex from $S$ with a vertex from $T$ in $G \backslash X$. An \stcut $X$ is said to be \textit{minimum} if there is no \stcut $X'$ such that $\card{X'} < \card{X}$.
For any \stcut $X$, its \textit{source side} $R(X,S)$ is the set of vertices that are reachable from $S$ in $G \backslash X$. Its \textit{target side} $R(X,T)$ contains the vertices reachable from $T$ in $G\backslash X$.
We define two sets $V^S(X)$ and $V^T(X)$:
\begin{itemize}
  \item set $V^S(X)= \set{u \in R(X,S);(u,v) \in X}$, {\em i.e.} the vertices of $R(X,S)$ incident to cut $X$,
  \item set $V^T(X)= \set{u \in R(X,T);(u,v) \in X}$, {\em i.e.} the vertices of $R(X,T)$ incident to cut $X$.
\end{itemize}

\textbf{Important and closest cuts.} As defined in~\cite{Ma06}, an \stcut $X$ is \textit{important} if there is no other \stcut $X'$ such that $\card{X'} \leq \card{X}$ and $R(X,S) \subsetneq R(X',S)$. 
Intuitively, an important \stcut is such that there is no other cut smaller in size which is closer to $T$. The number of important $(S,T)$-cuts of size at most $p$ depends only on $p$~\cite{ChLiLu09} and there is no more than one minimum important $(S,T)$-cut~\cite{Ma06}. Although the proofs in~\cite{Ma06} handle vertex cuts, an edge-to-vertex reduction preserves these properties on edge cuts~\cite{BoDaTh11,Ma06}.
  
\begin{lemma}[Unicity of minimum important cuts~\cite{Ma06}]
\label{le:impcutsenum}
  For disjoint sets of vertices $S$ and $T$, there is a unique minimum important $(S,T)$-cut and it may be found in polynomial time.
\end{lemma}

On undirected graphs, we say that an important $(T,S)$-cut is a \textit{closest} $(S,T)$-\textit{cut}. Fig.~\ref{fig:closest} gives an example of graph $G$ with two $(S,T)$-cuts $X_1$ and $X_2$, where $S = \set{s_1,s_2}$ and $T = \set{t}$.
\begin{figure}[h]
	\centering
	\scalebox{.60}{\begin{tikzpicture}[line cap=round,line join=round,>=triangle 45,x=1cm,y=1cm]
\coordinate (S1) at (1.0,2.5);
\coordinate (S2) at (1.0,3.5);
\coordinate (A) at (3.0,1.5);
\coordinate (B) at (3.0,2.5);
\coordinate (C) at (3.0,3.5);
\coordinate (D) at (3.0,4.5);
\coordinate (E) at (5.0,1.5);
\coordinate (F) at (5.0,2.5);
\coordinate (G) at (5.0,3.5);
\coordinate (H) at (7.0,1.5);
\coordinate (I) at (7.0,2.5);
\coordinate (J) at (7.0,3.5);
\coordinate (K) at (9.0,1.5);
\coordinate (L) at (9.0,2.5);
\coordinate (M) at (9.0,3.5);
\coordinate (T) at (11.0,3.0);
\draw [line width=1pt] (S1)--(A);
\draw [line width=1pt] (S1)--(B);
\draw [line width=1pt] (S2)--(C);
\draw [line width=1pt] (S2)--(D);
\draw [line width=1pt,color=blue] (A)--(E);
\draw [line width=1pt,color=blue] (A)--(F);
\draw [line width=1pt,color=blue] (B)--(E);
\draw [line width=1pt,color=blue] (C)--(F);
\draw [line width=1pt,color=blue] (D)--(G);
\draw [line width=1.2pt,color=red] (E)--(H);
\draw [line width=1.2pt,color=red] (F)--(I);
\draw [line width=1.2pt,color=red] (G)--(J);
\draw [line width=1pt] (H)--(K);
\draw [line width=1pt] (I)--(L);
\draw [line width=1pt] (J)--(M);
\draw [line width=1pt] (K)--(T);
\draw [line width=1pt] (L)--(T);
\draw [line width=1pt] (M)--(T);
\draw [fill=black] (S1) circle (2.5pt);
\draw [fill=black] (S2) circle (2.5pt);
\draw [fill=black] (A) circle (2.5pt);
\draw [fill=black] (B) circle (2.5pt);
\draw [fill=black] (C) circle (2.5pt);
\draw [fill=black] (D) circle (2.5pt);
\draw [fill=black] (E) circle (2.5pt);
\draw [fill=black] (F) circle (2.5pt);
\draw [fill=black] (G) circle (2.5pt);
\draw [fill=black] (H) circle (2.5pt);
\draw [fill=black] (I) circle (2.5pt);
\draw [fill=black] (J) circle (2.5pt);
\draw [fill=black] (K) circle (2.5pt);
\draw [fill=black] (L) circle (2.5pt);
\draw [fill=black] (M) circle (2.5pt);
\draw [fill=black] (T) circle (2.5pt);
\node [scale=1.2] at (0.7,3.8) {$s_2$};
\node [scale=1.2] at (0.7,2.1) {$s_1$};
\node [scale=1.2] at (11.2,3.4) {$t$};
\node [color=black,scale=1.2] at (2.0,1.3) {$Z_1$};
\node [color=blue,scale=1.2] at (4.0,1.0) {$X_1$};
\node [color=red,scale=1.2] at (6.0,1.0) {$X_2$};
\end{tikzpicture} }
	\caption{Illustration of Def.~\ref{def:closest} for closest $(S,T)$-cuts: $X_2$ is closest whereas $X_1$ is not.}
	\label{fig:closest}
\end{figure}
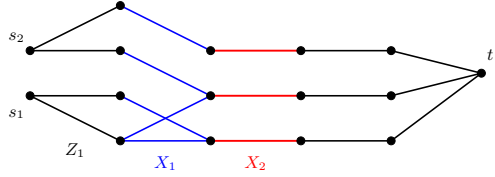
Cut $X_1$ is not closest as the edges incident to $S$ form a cut $Z_1$ smaller than $X_1$ and $R(Z_1,S) \subseteq R(X_1,S)$. Cut $X_2$ is closest because there is no cut with at most three edges whose reachable set of vertices is contained in $R(X_2,S)$.

\begin{definition}
  An \stcut $X$ is \textit{\anodyne}if there is no other \stcut $X'$ such that $\card{X'} \leq \card{X}$ and $R(X',S) \subseteq R(X,S)$.
  \label{def:closest}
\end{definition}

As a minimum closest $(S,T)$-cut is also a minimum important $(T,S)$-cut on undirected graphs,  there is a unique minimum closest $(S,T)$-cut according to Lemma~\ref{le:impcutsenum}.
Since the graph is uncapacitated, computing the minimum closest $(S,T)$-cut is made in time $O(mp)$, using $p$ iterations of Ford-Fulkerson's algorithm~\cite{FoFu56}.

\section{Framework: \network and Menger's paths} \label{sec:skeleton}

We introduce tools needed to design an algorithm solving \cmincuts\ in FPT$\langle p \rangle$ time, where $p$ is the size of any minimum $(S,T)$-cut.
We build the \textit{drainage}, a collection of minimum cuts $Z_i$, $i \in \set{1,\ldots,k}$, where $k < n$, such that at least one edge of any minimum $(S,T)$-cut $X$ belongs to $\bigcup_{i=1}^k Z_i$. Then, we highlight properties coming from Menger's theorem.

\subsection{Construction of the \network}

The \network $\mcalz\left(\mcali\right) = (Z_1,\ldots,Z_k)$ of an instance $\mcali = (G,S,T)$ is a collection of minimum $(S,T)$-cuts $Z_i$, $\card{Z_i} = p$, satisfying the following properties:
\begin{itemize}
\item there are less than $n$ cuts $Z_i$, {\em i.e.}  $1\le k < n$,
\item the reachable sets of cuts $Z_i$ fulfil $R(Z_i,S) \subsetneq R(Z_{i+1},S)$ for $i \in \set{1,\ldots,k-1}$,
\item for any minimum $(S,T)$-cut $X$, there is at least one cut $Z_i$ which has edges with $X$ in common: $X \cap Z_i \neq \emptyset$. 
\end{itemize}

We construct the \network iteratively.
Let $S_1 = S$ and $Z_1$ be the minimum \closest $(S_1,T)$-cut.
We fix $R_1 = R(Z_1,S)$. Let $S_2$ be the set of vertices incident to edges of $Z_1$ inside $R(Z_1,T)$: $S_2 = V^T(Z_1)=\set{v \notin R_1, (u,v) \in Z_1}$.

Next, we construct $Z_2$ which is the minimum closest $(S_2,T)$-cut in $G\backslash R(Z_1,S)$. If $\card{Z_2} > p$, the drainage construction stops. Otherwise, if $\card{Z_2} = p$, set $R_2$ follows the same scheme as $R_1$, $R_2=R(Z_2,S_2)$ in graph $G\backslash R(Z_1,S)$. We repeat the process until no more minimum $(S_i,T)$-cut $Z_i$ of size $p$ can be found. We denote by $k$ the number of cuts $Z_i$ produced and fix $R_{k+1} = R(Z_k,T)$. Cuts $Z_i$ form the \textit{minimum drainage cuts} of $\mcali$.

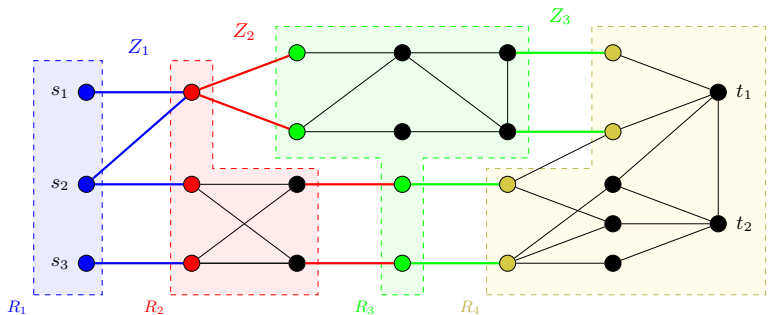
\begin{figure}[t]
\centering
\scalebox{.70}{\begin{tikzpicture}


\draw [dashed, color = blue, fill = white!92!blue] (1.0,0.4) -- (2.3,0.4) -- (2.3,4.85) -- (1.0,4.85) -- (1.0,0.4) node[below left] {$R_1$};
\draw [dashed, color = red, fill = white!92!red] (3.6,0.4) -- (6.4,0.4) -- (6.4,2.8) -- (4.4,2.8) -- (4.4,4.85) -- (3.6,4.85) -- (3.6,0.4) node[below left] {$R_2$};
\draw [dashed, color = green, fill = white!92!green] (7.6,0.4) -- (8.4,0.4) -- (8.4,3) -- (10.4,3) -- (10.4,5.5) -- (5.6,5.5) -- (5.6,3) -- (7.6,3) -- (7.6,0.4) node[below left] {$R_3$};
\draw [dashed, color = black!30!yellow, fill = white!90!yellow] (9.6,0.4) -- (14.9,0.4) -- (14.9,5.5) -- (11.6,5.5) -- (11.6,2.8) -- (9.6,2.8) -- (9.6,0.4) node[below left] {$R_4$};


\node[draw, circle, minimum height=0.2cm, minimum width=0.2cm, fill=blue] (P11) at (2*1,4.25) {};
\node[draw, circle, minimum height=0.2cm, minimum width=0.2cm, fill=blue] (P12) at (2*1,2.5) {};
\node[draw, circle, minimum height=0.2cm, minimum width=0.2cm, fill=blue] (P13) at (2*1,1) {};

\node[draw, circle, minimum height=0.2cm, minimum width=0.2cm, fill=red] (P21) at (2*2,4.25) {};
\node[draw, circle, minimum height=0.2cm, minimum width=0.2cm, fill=red] (P22) at (2*2,2.5) {};
\node[draw, circle, minimum height=0.2cm, minimum width=0.2cm, fill=red] (P23) at (2*2,1) {};

\node[draw, circle, minimum height=0.2cm, minimum width=0.2cm, fill=green] (P31) at (2*3,5) {};
\node[draw, circle, minimum height=0.2cm, minimum width=0.2cm, fill=green] (P32) at (2*3,3.5) {};
\node[draw, circle, minimum height=0.2cm, minimum width=0.2cm, fill=black] (P33) at (2*3,2.5) {};
\node[draw, circle, minimum height=0.2cm, minimum width=0.2cm, fill=black] (P34) at (2*3,1) {};

\node[draw, circle, minimum height=0.2cm, minimum width=0.2cm, fill=black] (P41) at (2*4,5) {};
\node[draw, circle, minimum height=0.2cm, minimum width=0.2cm, fill=black] (P42) at (2*4,3.5) {};
\node[draw, circle, minimum height=0.2cm, minimum width=0.2cm, fill=green] (P43) at (2*4,2.5) {};
\node[draw, circle, minimum height=0.2cm, minimum width=0.2cm, fill=green] (P44) at (2*4,1) {};

\node[draw, circle, minimum height=0.2cm, minimum width=0.2cm, fill=black] (P51) at (2*5,5) {};
\node[draw, circle, minimum height=0.2cm, minimum width=0.2cm, fill=black] (P52) at (2*5,3.5) {};
\node[draw, circle, minimum height=0.2cm, minimum width=0.2cm, fill=black!20!yellow] (P53) at (2*5,2.5) {};
\node[draw, circle, minimum height=0.2cm, minimum width=0.2cm, fill=black!20!yellow] (P54) at (2*5,1) {};

\node[draw, circle, minimum height=0.2cm, minimum width=0.2cm, fill=black!20!yellow] (P61) at (2*6,5) {};
\node[draw, circle, minimum height=0.2cm, minimum width=0.2cm, fill=black!20!yellow] (P62) at (2*6,3.5) {};
\node[draw, circle, minimum height=0.2cm, minimum width=0.2cm, fill=black] (P63) at (2*6,2.5) {};
\node[draw, circle, minimum height=0.2cm, minimum width=0.2cm, fill=black] (P64) at (2*6,1.75) {};
\node[draw, circle, minimum height=0.2cm, minimum width=0.2cm, fill=black] (P65) at (2*6,1) {};

\node[draw, circle, minimum height=0.2cm, minimum width=0.2cm, fill=black] (P71) at (2*7,4.25) {};
\node[draw, circle, minimum height=0.2cm, minimum width=0.2cm, fill=black] (P72) at (2*7,1.75) {};


\draw[line width = 1.2pt, color = blue] (P11) -- (P21);
\draw[line width = 1.2pt, color = red] (P21) -- (P31);
\draw (P31) -- (P41);
\draw (P41) -- (P51);
\draw[line width = 1.2pt, color = green] (P51) -- (P61);
\draw (P61) -- (P71);
\draw[line width = 1.2pt, color = blue] (P12) -- (P21);
\draw[line width = 1.2pt, color = red] (P21) -- (P32);
\draw (P32) -- (P42);
\draw (P42) -- (P52);
\draw[line width = 1.2pt, color = green] (P52) -- (P62);
\draw (P62) -- (P71);
\draw[line width = 1.2pt, color = blue] (P12) -- (P22);
\draw (P22) -- (P33);
\draw (P23) -- (P33);
\draw[line width = 1.2pt, color = red] (P33) -- (P43);
\draw[line width = 1.2pt, color = green] (P43) -- (P53);
\draw (P53) -- (P64);
\draw (P64) -- (P54);
\draw (P54) -- (P65);
\draw (P65) -- (P72);
\draw[line width = 1.2pt, color = blue] (P13) -- (P23);
\draw (P23) -- (P34);
\draw[line width = 1.2pt, color = red] (P34) -- (P44);
\draw[line width = 1.2pt, color = green] (P44) -- (P54);
\draw (P54) -- (P63);
\draw (P63) -- (P72);

\draw (P22) -- (P34);
\draw (P32) -- (P41);
\draw (P41) -- (P52);
\draw (P52) -- (P51);
\draw (P53) -- (P62);
\draw (P23) -- (P34);
\draw (P64) -- (P72);
\draw (P63) -- (P71);
\draw (P71) -- (P72);


\node[scale=1.2] at (1.5,4.25) {$s_1$};
\node[scale=1.2] at (1.5,2.5) {$s_2$};
\node[scale=1.2] at (1.5,1) {$s_3$};

\node[scale=1.2] at (14.5,4.25) {$t_1$};
\node[scale=1.2] at (14.5,1.75) {$t_2$};

\node[scale=1.2, color = blue] at (3.0,5.1) {$Z_1$};
\node[scale=1.2, color = red] at (5.0,5.4) {$Z_2$};
\node[scale=1.2, color = green] at (11.0,5.7) {$Z_3$};

\end{tikzpicture}}
\caption{The drainage (cuts $Z_i$, sets $R_i$ and $S_i$) for an instance containing graph $G$, sources $S = \set{s_1,s_2,s_3}$ and targets $T = \set{t_1,t_2}$. Here, $R_1 = S_1$ (in general, $R_1 \supseteq S_1$).}
\label{fig:skeleton}
\end{figure}

Fig.~\ref{fig:skeleton} provides us with an example of graph $G$ with $S = \set{s_1,s_2,s_3}$ and $T = \set{t_1,t_2}$ and indicates its drainage. The size of minimum $(S,T)$-cuts is $p=4$. Blue, red, and green edges represent minimum drainage cuts $Z_1$, $Z_2$, and $Z_3$, respectively. Similarly, blue, red, green, and yellow vertices represent sets $S_1=S$, $S_2$, $S_3$, and $S_4$. Reachable sets $R_1$, $R_2$, $R_3$, and $R_4$ are also appropriately colored. As the size of the minimum cut between $S_4$ (yellow vertices) and $T$ in graph $G\backslash R(Z_3,S)$ is greater than $p$, we have $k=3$. 

We emphasize that set $R_i$, which is $R(Z_i,S_i)$ taken in $G\backslash R(Z_{i-1},S)$, and set $R(Z_i,S)$ are different for $i \neq 1$. On the one hand, set $R(Z_i,S) = \bigcup_{\ell = 1}^{i} R_{\ell}$ contains the vertices reachable from $S$ in graph $G$ deprived of $Z_i$. On the other hand, set $R_i$ can be written $R_i = R(Z_i,S) \backslash R(Z_{i-1},S)$. Sets $R_i$ and $R_{i+1}$ are disjoint and nonempty, as $S_i \subseteq R_i$ and $S_{i+1} \subseteq R_{i+1}$. Moreover, the minimum drainage cuts are disjoint: $Z_i \cap Z_j = \emptyset$. The number $k$ of minimum drainage cuts is less than $n$ and the running time needed to construct the drainage is in $O(mnp)$. The reachable vertex sets of cuts $Z_i$ are included one into another: $R(Z_i,S) \subsetneq R(Z_{i+1},S)$. The following theorem shows that, for any minimum $(S,T)$-cut $X$, there is a cut $Z_i$ containing edges of $X$. Among cuts $Z_i$ sharing edges with $X$, we are interested in the one with the smallest index.

\begin{definition}[Front of $X$]
Front of $X$, $i(X)$, $1\le i(X) \le k$ is the smallest index $i$ such that $Z_i \cap X \neq \emptyset$.
\label{def:front}
\end{definition}

The next theorem states the properties of $i(X)$ for any minimum $(S,T)$-cut.

\begin{theorem}
Any minimum $(S,T)$-cut $X$ admits a front $i(X)$ and\\ $X \cap E\left[R(Z_{i(X)},S)\right] = \emptyset$.
\label{th:skeleton}
\end{theorem}
\begin{proof}
First, cut $X$ cannot be entirely included in $E\left[R_{k+1}\right]$. If it was, it would be a minimum $(S_{k+1},T)$-cut of size $p$, which contradicts the drainage definition. So, some edges of $X$ are incident to $R(Z_k,S)$. 

Second, no edge of $X$ belongs to $E\left[R(Z_1,S)\right]$ as cut $Z_1$ is the minimum closest $(S,T)$-cut. Therefore, there is an index $i \geq 1$ such that no edge of $X$ belongs to $E\left[R(Z_{i},S)\right]$ but at least one has an endpoint in $R_{i+1}$. 

Obviously, if an edge of $X$ belongs to $Z_{i}$, the theorem holds: $i = i(X)$.
We study the case where no edge of $X$ belongs to $Z_{i}$ and there is an edge $e$ of $X$, $e \in E\left[R_{i+1}\right]$.
According to the definition of index $i$, no edge of $X$ belongs to $E\left[R(Z_{i},S)\right]$, and therefore all edges of $X$ have to be on the target side $E\left[R(Z_{i},T)\right]$. Therefore, $X$ is a minimum $(S_{i+1},T)$-cut in graph $G\backslash R(Z_{i},S)$. Either cut $X$ is a minimum closest $(S_{i+1},T)$-cut (and then we fix $Z=X$) or the minimum closest $(S_{i+1},T)$-cut $Z$ is different than $X$ and it fulfils $R(Z,S_{i+1}) \subsetneq R(X,S_{i+1})$. Since there is an edge $e\in X\cap E\left[R_{i+1}\right]$, then one of its endpoint $v\in R_{i+1}$ belongs to $R(X,T)$. Consequently, $v \notin R(Z,S_{i+1})$. This brings a contradiction: cutset $Z_{i+1}$ is the unique minimum closest $(S_{i+1},T)$-cut and $R(Z_{i+1},S_{i+1}) = R_{i+1}$. As vertex $v$ can be reached from $S_{i+1}$ after the removal of $Z_{i+1}$ but not after the removal of $Z$, cuts $Z$ and $Z_{i+1}$ differ. Thus, cut $Z$ cannot be the minimum closest $(S_{i+1},T)$-cut.

In summary, there is an index $i$ such that no edge of $X$ belongs to $E\left[R(Z_{i},S)\right]$ and, moreover, $X \cap Z_i \neq \emptyset$. This means that there is no index $\ell < i$ such that $X \cap Z_{\ell} \neq \emptyset$. Consequently, index $i$ is the front of $X$: $i = i(X)$.
\end{proof}

The reader can verify that any minimum $(S,T)$-cut of $G$ contains some edges of at least one cut $Z_1$, $Z_2$, or $Z_3$ in Fig.~\ref{fig:skeleton}.

\subsection{Menger's paths} \label{subsec:menger}

Menger's theorem states that the size of the minimum edge $(S,T)$-cuts is equal to the maximum number of edge-disjoint $(S,T)$-paths~\cite{Me27}. One of these largest sets of edge-disjoint $(S,T)$-paths can be found in polynomial time using flow techniques~\cite{FoFu56}. We denote by $\mcalq = \set{Q_1,\ldots,Q_p}$ such a set of $p$ edge-disjoint $(S,T)$-paths, taken arbitrarily. We call paths from $\mcalq$ \textit{Menger's paths}, to distinguish them from other paths in graph $G$.

Set $\mcalq$ is used to identify minimum $(S,T)$-cuts. It is fixed throughout the course of the proofs in this article. The observation that edges of all minimum $(S,T)$-cuts belong to the paths from $\mcalq$ is formulated in:

\begin{lemma}
For any minimum $(S,T)$-cut $X$, each Menger's path $Q_j$ contains one edge of $X$. If $Q_j : v_1^{(j)} \cdot v_2^{(j)} \cdots v_{\ell}^{(j)} \cdot v_{\ell+1}^{(j)} \cdots$  and $(v_{\ell}^{(j)},v_{\ell +1}^{(j)}) \in X$, then $v_{\ell}^{(j)} \in R(X,S)$ and $v_{\ell+1}^{(j)} \in R(X,T)$.
\label{le:menger}
\end{lemma}
\begin{proof}[Proof of Lemma~\ref{le:menger}]
If a Menger's path $Q_j$ did not contain any edge of cut $X$, then set $X$ would not be a $(S,T)$-cut. Similarly, if path $Q_j$ contained  at least two edges of cut $X$, then another Menger's path would not contain any edge of $X$ as $\card{\mcalq} = \card{X}$.

We suppose that $(v_{\ell}^{(j)},v_{\ell +1}^{(j)})$ is the edge of $X$ in path $Q_j$. To prove that $v_{\ell}^{(j)} \in R(X,S)$ and $v_{\ell+1}^{(j)} \in R(X,T)$, we admit the contrary. 
As $Q_j$ contains one edge of $X$, the segment $v_1^{(j)}\cdot v_2^{(j)}\cdot \ldots\cdot v_{\ell}^{(j)}$ is open in $G \backslash X$ between $v_1^{(j)} \in S$ and $v_{\ell}^{(j)} \in R(X,T)$. Therefore, $X$ is not an $(S,T)$-cut as at least one vertex of $R(X,S)$ is connected to $R(X,T)$, a contradiction.
\end{proof}

Its consequence is that each edge of a cut $Z_i$ belongs to a Menger's path.
Fig.~\ref{fig:mengerspaths} illustrates the Menger's paths on the instance $(G,S,T)$ already introduced in Fig.~\ref{fig:skeleton}. As the minimum $(S,T)$-cut size is four, there are four edge-disjoint $(S,T)$-paths distinguished with colors. Menger's paths are naturally oriented from sources to targets.

\section{Counting minimum edge $(S,T)$-cuts in undirected graphs} \label{sec:algo}

We describe our algorithm which counts all minimum $(S,T)$-cuts in an undirected graph $G$. Based on the concepts introduced in Section~\ref{sec:skeleton}, we prove not only that it achieves its objective but also that its time complexity is FPT.

\begin{figure}[t]
\centering
\scalebox{.7}{\begin{tikzpicture}


\draw [dashed, color = white!40!gray, fill = white!92!gray] (1.0,0.4) -- (2.3,0.4) -- (2.3,4.85) -- (1.0,4.85) -- (1.0,0.4) node[below left] {$R_1$};
\draw [dashed, color = white!40!gray, fill = white!92!gray] (3.6,0.4) -- (6.4,0.4) -- (6.4,2.8) -- (4.4,2.8) -- (4.4,4.85) -- (3.6,4.85) -- (3.6,0.4) node[below left] {$R_2$};
\draw [dashed, color = white!40!gray, fill = white!92!gray] (7.6,0.4) -- (8.4,0.4) -- (8.4,3) -- (10.4,3) -- (10.4,5.5) -- (5.6,5.5) -- (5.6,3) -- (7.6,3) -- (7.6,0.4) node[below left] {$R_3$};
\draw [dashed, color = white!40!gray, fill = white!92!gray] (9.6,0.4) -- (14.9,0.4) -- (14.9,5.5) -- (11.6,5.5) -- (11.6,2.8) -- (9.6,2.8) -- (9.6,0.4) node[below left] {$R_4$};


\node[draw, circle, minimum height=0.2cm, minimum width=0.2cm, fill=black] (P11) at (2*1,4.25) {};
\node[draw, circle, minimum height=0.2cm, minimum width=0.2cm, fill=black] (P12) at (2*1,2.5) {};
\node[draw, circle, minimum height=0.2cm, minimum width=0.2cm, fill=black] (P13) at (2*1,1) {};

\node[draw, circle, minimum height=0.2cm, minimum width=0.2cm, fill=black] (P21) at (2*2,4.25) {};
\node[draw, circle, minimum height=0.2cm, minimum width=0.2cm, fill=black] (P22) at (2*2,2.5) {};
\node[draw, circle, minimum height=0.2cm, minimum width=0.2cm, fill=black] (P23) at (2*2,1) {};

\node[draw, circle, minimum height=0.2cm, minimum width=0.2cm, fill=black] (P31) at (2*3,5) {};
\node[draw, circle, minimum height=0.2cm, minimum width=0.2cm, fill=black] (P32) at (2*3,3.5) {};
\node[draw, circle, minimum height=0.2cm, minimum width=0.2cm, fill=black] (P33) at (2*3,2.5) {};
\node[draw, circle, minimum height=0.2cm, minimum width=0.2cm, fill=black] (P34) at (2*3,1) {};

\node[draw, circle, minimum height=0.2cm, minimum width=0.2cm, fill=black] (P41) at (2*4,5) {};
\node[draw, circle, minimum height=0.2cm, minimum width=0.2cm, fill=black] (P42) at (2*4,3.5) {};
\node[draw, circle, minimum height=0.2cm, minimum width=0.2cm, fill=black] (P43) at (2*4,2.5) {};
\node[draw, circle, minimum height=0.2cm, minimum width=0.2cm, fill=black] (P44) at (2*4,1) {};

\node[draw, circle, minimum height=0.2cm, minimum width=0.2cm, fill=black] (P51) at (2*5,5) {};
\node[draw, circle, minimum height=0.2cm, minimum width=0.2cm, fill=black] (P52) at (2*5,3.5) {};
\node[draw, circle, minimum height=0.2cm, minimum width=0.2cm, fill=black] (P53) at (2*5,2.5) {};
\node[draw, circle, minimum height=0.2cm, minimum width=0.2cm, fill=black] (P54) at (2*5,1) {};

\node[draw, circle, minimum height=0.2cm, minimum width=0.2cm, fill=black] (P61) at (2*6,5) {};
\node[draw, circle, minimum height=0.2cm, minimum width=0.2cm, fill=black] (P62) at (2*6,3.5) {};
\node[draw, circle, minimum height=0.2cm, minimum width=0.2cm, fill=black] (P63) at (2*6,2.5) {};
\node[draw, circle, minimum height=0.2cm, minimum width=0.2cm, fill=black] (P64) at (2*6,1.75) {};
\node[draw, circle, minimum height=0.2cm, minimum width=0.2cm, fill=black] (P65) at (2*6,1) {};

\node[draw, circle, minimum height=0.2cm, minimum width=0.2cm, fill=black] (P71) at (2*7,4.25) {};
\node[draw, circle, minimum height=0.2cm, minimum width=0.2cm, fill=black] (P72) at (2*7,1.75) {};


\draw[->,line width = 1.4pt, color = brown] (P11) -- (P21);
\draw[->,line width = 1.4pt, color = brown] (P21) -- (P31);
\draw[->,line width = 1.4pt, color = brown] (P31) -- (P41);
\draw[->,line width = 1.4pt, color = brown] (P41) -- (P51);
\draw[->,line width = 1.4pt, color = brown] (P51) -- (P61);
\draw[->,line width = 1.4pt, color = brown] (P61) -- (P71);
\draw[->,line width = 1.4pt, color = purple] (P12) -- (P21);
\draw[->,line width = 1.4pt, color = purple] (P21) -- (P32);
\draw (P32) -- (P42);
\draw (P42) -- (P52);
\draw[->,line width = 1.4pt, color = purple] (P52) -- (P62);
\draw[->,line width = 1.4pt, color = purple] (P62) -- (P71);
\draw[->,line width = 1.4pt, color = orange] (P12) -- (P22);
\draw[->,line width = 1.4pt, color = orange] (P22) -- (P33);
\draw[->,line width = 1.4pt, color = orange] (P33) -- (P43);
\draw[->,line width = 1.4pt, color = orange] (P43) -- (P53);
\draw[->,line width = 1.4pt, color = orange] (P53) -- (P64);
\draw[->,line width = 1.4pt, color = orange] (P64) -- (P54);
\draw[->,line width = 1.4pt, color = orange] (P54) -- (P65);
\draw[->,line width = 1.4pt, color = orange] (P65) -- (P72);
\draw[->,line width = 1.4pt, color = black!50!green]  (P13) -- (P23);
\draw[->,line width = 1.4pt, color = black!50!green] (P23) -- (P34);
\draw[->,line width = 1.4pt, color = black!50!green] (P34) -- (P44);
\draw[->,line width = 1.4pt, color = black!50!green] (P44) -- (P54);
\draw[->,line width = 1.4pt, color = black!50!green] (P54) -- (P63);
\draw[->,line width = 1.4pt, color = black!50!green] (P63) -- (P72);

\draw (P23) -- (P33);
\draw (P22) -- (P34);
\draw[->,line width = 1.4pt, color = purple] (P32) -- (P41);
\draw[->,line width = 1.4pt, color = purple] (P41) -- (P52);
\draw (P52) -- (P51);
\draw (P53) -- (P62);
\draw (P64) -- (P72);
\draw (P63) -- (P71);
\draw (P71) -- (P72);


\node[scale=1.2] at (1.5,4.25) {$s_1$};
\node[scale=1.2] at (1.5,2.5) {$s_2$};
\node[scale=1.2] at (1.5,1) {$s_3$};

\node[scale=1.2] at (14.5,4.25) {$t_1$};
\node[scale=1.2] at (14.5,1.75) {$t_2$};


\end{tikzpicture}}
\caption{Menger's paths in graph $G$ with sources $S = \set{s_1,s_2,s_3}$, targets $T = \set{t_1,t_2}$.} 
\label{fig:mengerspaths}
\end{figure}
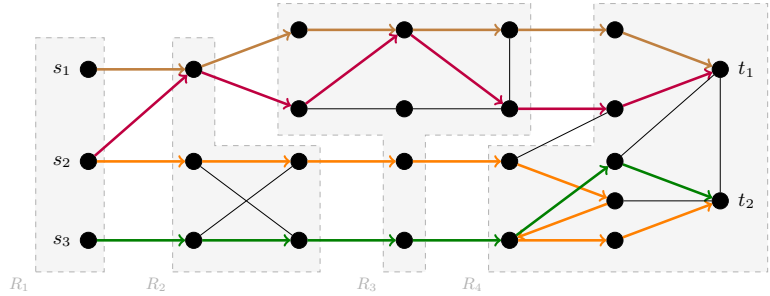

\subsection{Dams and dry areas}

We begin by the definition of \textit{dams} which are subsets of cuts $Z_i$ of the drainage of $G$.

\begin{definition}[Dam]
A dam $B_i$ is a nonempty subset of a minimum drainage cut $Z_i$, {\em i.e.} $B_i \subseteq Z_i$, $B_i \neq \emptyset$.
\label{def:dams}
\end{definition}

Thanks to this definition, Theorem~\ref{th:skeleton} together with the concept of the front makes us observe that any minimum $(S,T)$-cut $X$ contains a \textit{front dam}:
\begin{definition}[Front dam]
The \textit{front dam} of a minimum $(S,T)$-cut $X$ is $B_{i(X)} = X\cap Z_{i(X)}$. 
\label{def:front_dam}
\end{definition}
We know that all edges in $X \backslash B_{i(X)}$ belong to the target side of $Z_{i(X)}$, $E\left[R(Z_{i(X)},T)\right]$, and the source side of $Z_{i(X)}$ is empty, $X \cap E\left[R(Z_{i(X)},S)\right] = \emptyset$. 
If $X\backslash B_{i(X)} = \emptyset$, then $X = Z_{i(X)}$. 
A dam $B_i$ is characterized by:
\begin{itemize}
\item its \textit{level}, {\em i.e.} the index $i$ of the cut $Z_i$ it belongs to,
\item its \textit{signature} $\sigma(B_i) = \set{Q_j : B_i \cap Q_j \neq \emptyset}$, {\em i.e.} the set of Menger's paths passing through it.
\end{itemize}

Choking graph $G$ with dam $B_{i(X)}$ puts in evidence a subgraph which still connects $S$ and $T$ through $X\backslash B_{i(X)}$. Our idea is to dam a graph gradually in order to dry it completely.

The description of the method we devised to reach this goal requires a transformation of $G$ into $G_{D}$ which is actually $G$ with certain edges directed ($G_{D}$ is a mixed graph). 
If edge $e$ does not belong to a Menger's path, it stays undirected. For path $Q_j : v_1^{(j)}\cdot v_2^{(j)}\cdot v_3^{(j)}\cdots$
, edges $(v_i^{(j)},v_{i+1}^{(j)})$ become arcs $(v_i^{(j)},v_{i+1}^{(j)})$, respecting the natural flow from sources to targets.

Fig.~\ref{fig:mengerspaths} illustrates graph $G_{D}$. Arrows indicate the arcs while bare segments represent its edges. According to Lemma~\ref{le:menger}, any minimum $(S,T)$-cut of $G$ is made up of arcs in $G_{D}$. Minimum drainage cuts $Z_i$ are thus composed of arcs, directed from $R_i$ to $R_{i+1}$. We insist on the fact that graph $G_{D}$ is only used to define the notion of dry area, we do not count minimum cuts in it.

\begin{definition}[Dry area]
The dry area of $B_i$ is the set $A^*(B_i)$ which contains the vertices of $G$ which are not reachable from $S$ in graph $G_{D}$ deprived of $B_i$, {\em i.e.} $G_{D}\backslash B_i$. 
\label{def:dry_area}
\end{definition}

In a less formal way, set $A^*(B_i)$ keeps vertices which are dried as $B_i$ is the only means to irrigate them. The definition of the dry instance follows.

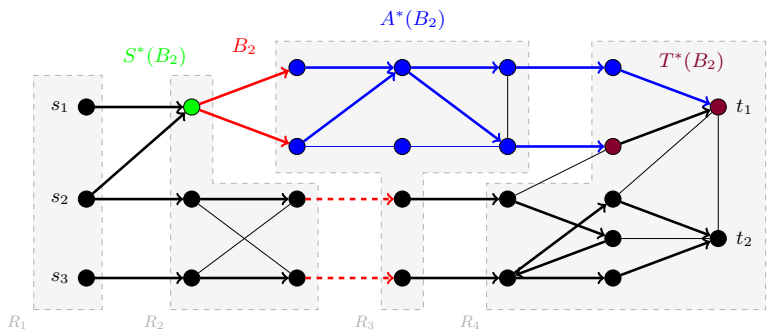
\begin{figure}[h]
\centering
\scalebox{.70}{\begin{tikzpicture}


\draw [dashed, color = white!40!gray, fill = white!92!gray] (1.0,0.4) -- (2.3,0.4) -- (2.3,4.85) -- (1.0,4.85) -- (1.0,0.4) node[below left] {$R_1$};
\draw [dashed, color = white!40!gray, fill = white!92!gray] (3.6,0.4) -- (6.4,0.4) -- (6.4,2.8) -- (4.4,2.8) -- (4.4,4.85) -- (3.6,4.85) -- (3.6,0.4) node[below left] {$R_2$};
\draw [dashed, color = white!40!gray, fill = white!92!gray] (7.6,0.4) -- (8.4,0.4) -- (8.4,3) -- (10.4,3) -- (10.4,5.5) -- (5.6,5.5) -- (5.6,3) -- (7.6,3) -- (7.6,0.4) node[below left] {$R_3$};
\draw [dashed, color = white!40!gray, fill = white!92!gray] (9.6,0.4) -- (14.9,0.4) -- (14.9,5.5) -- (11.6,5.5) -- (11.6,2.8) -- (9.6,2.8) -- (9.6,0.4) node[below left] {$R_4$};


\node[draw, circle, minimum height=0.2cm, minimum width=0.2cm, fill=black] (P11) at (2*1,4.25) {};
\node[draw, circle, minimum height=0.2cm, minimum width=0.2cm, fill=black] (P12) at (2*1,2.5) {};
\node[draw, circle, minimum height=0.2cm, minimum width=0.2cm, fill=black] (P13) at (2*1,1) {};

\node[draw, circle, minimum height=0.2cm, minimum width=0.2cm, fill=green] (P21) at (2*2,4.25) {};
\node[draw, circle, minimum height=0.2cm, minimum width=0.2cm, fill=black] (P22) at (2*2,2.5) {};
\node[draw, circle, minimum height=0.2cm, minimum width=0.2cm, fill=black] (P23) at (2*2,1) {};

\node[draw, circle, minimum height=0.2cm, minimum width=0.2cm, fill=blue] (P31) at (2*3,5) {};
\node[draw, circle, minimum height=0.2cm, minimum width=0.2cm, fill=blue] (P32) at (2*3,3.5) {};
\node[draw, circle, minimum height=0.2cm, minimum width=0.2cm, fill=black] (P33) at (2*3,2.5) {};
\node[draw, circle, minimum height=0.2cm, minimum width=0.2cm, fill=black] (P34) at (2*3,1) {};

\node[draw, circle, minimum height=0.2cm, minimum width=0.2cm, fill=blue] (P41) at (2*4,5) {};
\node[draw, circle, minimum height=0.2cm, minimum width=0.2cm, fill=blue] (P42) at (2*4,3.5) {};
\node[draw, circle, minimum height=0.2cm, minimum width=0.2cm, fill=black] (P43) at (2*4,2.5) {};
\node[draw, circle, minimum height=0.2cm, minimum width=0.2cm, fill=black] (P44) at (2*4,1) {};

\node[draw, circle, minimum height=0.2cm, minimum width=0.2cm, fill=blue] (P51) at (2*5,5) {};
\node[draw, circle, minimum height=0.2cm, minimum width=0.2cm, fill=blue] (P52) at (2*5,3.5) {};
\node[draw, circle, minimum height=0.2cm, minimum width=0.2cm, fill=black] (P53) at (2*5,2.5) {};
\node[draw, circle, minimum height=0.2cm, minimum width=0.2cm, fill=black] (P54) at (2*5,1) {};

\node[draw, circle, minimum height=0.2cm, minimum width=0.2cm, fill=blue] (P61) at (2*6,5) {};
\node[draw, circle, minimum height=0.2cm, minimum width=0.2cm, fill=black!30!purple] (P62) at (2*6,3.5) {};
\node[draw, circle, minimum height=0.2cm, minimum width=0.2cm, fill=black] (P63) at (2*6,2.5) {};
\node[draw, circle, minimum height=0.2cm, minimum width=0.2cm, fill=black] (P64) at (2*6,1.75) {};
\node[draw, circle, minimum height=0.2cm, minimum width=0.2cm, fill=black] (P65) at (2*6,1) {};

\node[draw, circle, minimum height=0.2cm, minimum width=0.2cm, fill=black!30!purple] (P71) at (2*7,4.25) {};
\node[draw, circle, minimum height=0.2cm, minimum width=0.2cm, fill=black] (P72) at (2*7,1.75) {};


\draw[->,line width = 1.4pt, color = black] (P11) -- (P21);
\draw[->,line width = 1.4pt, color = red] (P21) -- (P31);
\draw[->,line width = 1.4pt, color = blue] (P31) -- (P41);
\draw[->,line width = 1.4pt, color = blue] (P41) -- (P51);
\draw[->,line width = 1.4pt, color = blue] (P51) -- (P61);
\draw[->,line width = 1.4pt, color = blue] (P61) -- (P71);
\draw[->,line width = 1.4pt, color = black] (P12) -- (P21);
\draw[->,line width = 1.4pt, color = red] (P21) -- (P32);
\draw[color=blue] (P32) -- (P42);
\draw[color=blue] (P42) -- (P52);
\draw[->,line width = 1.4pt, color = blue] (P52) -- (P62);
\draw[->,line width = 1.4pt, color = black] (P62) -- (P71);
\draw[->,line width = 1.4pt, color = black] (P12) -- (P22);
\draw[->,line width = 1.4pt, color = black] (P22) -- (P33);
\draw[->,line width = 1.4pt, color = red, dashed] (P33) -- (P43);
\draw[->,line width = 1.4pt, color = black] (P43) -- (P53);
\draw[->,line width = 1.4pt, color = black] (P53) -- (P64);
\draw[->,line width = 1.4pt, color = black] (P64) -- (P54);
\draw[->,line width = 1.4pt, color = black] (P54) -- (P65);
\draw[->,line width = 1.4pt, color = black] (P65) -- (P72);
\draw[->,line width = 1.4pt, color = black]  (P13) -- (P23);
\draw[->,line width = 1.4pt, color = black] (P23) -- (P34);
\draw[->,line width = 1.4pt, color = red, dashed] (P34) -- (P44);
\draw[->,line width = 1.4pt, color = black] (P44) -- (P54);
\draw[->,line width = 1.4pt, color = black] (P54) -- (P63);
\draw[->,line width = 1.4pt, color = black] (P63) -- (P72);

\draw (P23) -- (P33);
\draw (P22) -- (P34);
\draw[->,line width = 1.4pt, color = blue] (P32) -- (P41);
\draw[->,line width = 1.4pt, color = blue] (P41) -- (P52);
\draw (P52) -- (P51);
\draw (P53) -- (P62);
\draw (P23) -- (P34);
\draw (P64) -- (P72);
\draw (P63) -- (P71);
\draw (P71) -- (P72);


\node[scale=1.2] at (1.5,4.25) {$s_1$};
\node[scale=1.2] at (1.5,2.5) {$s_2$};
\node[scale=1.2] at (1.5,1) {$s_3$};

\node[scale=1.2] at (14.5,4.25) {$t_1$};
\node[scale=1.2] at (14.5,1.75) {$t_2$};

\node[scale=1.2, color = red] at (5.0,5.4) {$B_2$};
\node[scale=1.2, color = green] at (3.3,5.2) {$S^*(B_2)$};
\node[scale=1.2, color = blue] at (8.2,5.9) {$A^*(B_2)$};
\node[scale=1.2, color = black!30!purple] at (13.5,5.1) {$T^*(B_2)$};

\end{tikzpicture}}
\caption{An example of dam $B_2$ and its dry instance $\mcald\left(\mcali,B_2\right) = \left(G^*(B_2),S^*(B_2),T^*(B_2)\right)$}
\label{fig:dry_area}
\end{figure}

\begin{definition}[Dry instance] The dry instance induced by a dam $B_i$ is an instance $\mcald\left(\mcali,B_i\right) =$ $\left(G^*(B_i),S^*(B_i),T^*(B_i)\right)$ with graph $G^*(B_i)= \left(V^*(B_i),E^*(B_i)\right)$. In particular,
\begin{itemize}
\item set $S^*(B_i)$ keeps vertices reachable from $S$ ``just before'' dam $B_i$. Formally, it contains the tails of arcs in $B_i$: $S^*(B_i) = \set{u : (u,v) \in B_i}$,
\item set $T^*(B_i)$ keeps vertices placed ``after'' dam $B_i$ which become irrigated in $G_{D}\backslash B_i$. Formally, it contains the heads of arcs which have their tail either inside $S^*(B_i)$ or inside $A^*(B_i)$ and their head outside: $T^*(B_i) = \set{v \notin A^*(B_i): (u,v) \in E, u \in S^*(B_i) \cup A^*(B_i)}$,
\item set $V^*(B_i)$ is the union: $V^*(B_i) = S^*(B_i) \cup A^*(B_i) \cup T^*(B_i)$,
\item set $E^*(B_i)$ stores edges of $G$ which lie inside the dry area of $B_i$ or on its border (one endpoint is outside) in $G_D$. Formally, it is composed of edges with two endpoints in $V^*(B_i)$ and at least one of them in $A^*(B_i)$: $E^*(B_i) = \set{(u,v) \in E: u \in A^*(B_i), v \in V^*(B_i)}$.
\end{itemize}
\label{def:dry_instance}
\end{definition}

Fig.~\ref{fig:dry_area} gives an example of dam $B_2 \subseteq Z_2$ and the dry instance it induces in $G$. Its arcs are drawn in red, arcs of $Z_2\backslash B_2$ are red and dashed. Blue vertices represent the vertices unreachable from $S$ in $G_{D} \backslash B_2$, {\em i.e.} set $A^*(B_2)$. Sets $S^*(B_2)$ and $T^*(B_2)$ are drawn in green and purple, respectively. Set $E^*(B_2)$ is composed of dam $B_2$ (red arcs) and blue edges/arcs.

An important property of dry areas is that there is no arc $(u,v)$ of $G_D$ ``entering'' in the dry area $A^*(B_i)$, except for arcs in $B_i$.

\begin{lemma}
For any dam $B_i$, there is no arc $(u,v)$ in $G_D$ such that $u \notin A^*(B_i)$ and $v \in A^*(B_i)$, except for arcs in $B_i$. Moreover, there is no undirected edge with exactly one endpoint in $A^*(B_i)$.
\label{le:enter_dry}
\end{lemma}
\begin{proof}
Suppose that such an arc $(u,v) \notin B_i$ exists. As $u \notin A^*(B_i)$, it is reachable from $S$ in $G_D\backslash B_i$. Therefore, $v$ can be reached too: this contradicts $v \in A^*(B_i)$. For the same reason, there is no undirected edge $(u,v)$ with only endpoint in $A^*(B_i)$ as the existence of this edge makes both its endpoints be reachable from $S$ in $G_D\backslash B_i$.
\end{proof}

In Theorem~\ref{th:keystone}, we provide a characterization of any minimum $(S,T)$-cut which is based on dry instances and on \textit{closest dams}. We start by:

\begin{definition}
A dam $B_h$ is \textit{closer} than dam $B_i$ if (i) $h < i$, (ii) $\sigma(B_h) = \sigma(B_i)$, and (iii) edges in $B_i$ are the only edges of level $i$ inside the dry instance of $B_h$: $E^*(B_h) \cap Z_i = B_i$.
\label{def:closer_dam}
\end{definition}

As a consequence, the dry area of $B_i$ is included in the dry area of $B_h$ when $B_h$ is closer than $B_i$: $A^*(B_i) \subsetneq A^*(B_h)$. Indeed, if a vertex is unreachable from $S$ in $G_{D}\backslash B_i$,  then it also is unreachable from $S$ in $G_{D}\backslash B_h$ as arcs of $B_i$ cannot be attained from $S$ in $G_{D}\backslash B_h$ according to Def.~\ref{def:closer_dam}.

\begin{definition}[Closest dam]
Dam $B_i$ is a closest dam if no dam $B_h$, $h < i$ is closer than $B_i$.
\label{def:closest_dam}
\end{definition}

For any dam $B_i$, either $B_i$ is a closest dam or there is a closest dam $B_h\neq B_i$, closer than $B_i$. Each dam $B_i$ admits a closest dam (itself or $B_h$) which is unique.

\begin{lemma}
Any dam $B_i$ has a unique closest dam.
\label{le:unique_closest}
\end{lemma}
\begin{proof}
If $B_i$ is already a closest dam, then it is its own unique closest dam. 

Now, we suppose that there are two closest dams of $B_i$, denoted by $B_{h_1}$ and $B_{h_2}$. Necessarily, $h_1 \neq h_2$, otherwise $B_{h_1} = B_{h_2}$ as, according to Def.~\ref{def:closer_dam}, they have the same signature. 

We prove that under the hypothesis $h_1 < h_2$, dam $B_{h_1}$ is closer than $B_{h_2}$.
Suppose, towards a contradiction, that there is an arc $e_2=(u_2,v_2) \in B_{h_2}$ which does not belong to set $E^*(B_{h_1})$. As $e_2 \notin E^*(B_{h_1})$,  vertex $u_2$ is not inside $A^*(B_{h_1})$, otherwise, according to Def.~\ref{def:closer_dam}, $e_2$ would belong to $E^*(B_{h_1})$. So, there is a path $\widehat{Q}$ connecting sources from $S$ with $u_2$, which avoids arcs in $B_{h_1}$. Arc $e_2$ belongs to a Menger's path $Q_j$. A section of this path, denoted by $\widehat{Q}_j'$, connects $u_2$ with a tail $u_3$ of an arc $(u_3,v_3)$ in $B_i$, because $Q_j$ passes through $B_i$: $ \sigma(B_i) = \sigma(B_{h_2})$. The concatenated path $\widehat{Q} \cdot \widehat{Q}_j'$ connects $S$ with $B_i$ while avoiding $B_{h_1}$: this is a contradiction as $B_i \subseteq E^*(B_{h_1})$. Consequently, $B_{h_2} \subseteq E^*(B_{h_1}) \cap Z_{h_2}$.

The equality $B_{h_2} = E^*(B_{h_1}) \cap Z_{h_2}$ comes from the fact that $B_{h_1}$ and $B_{h_2}$ have the same signature $\sigma(B_i)$. Arcs in dam $E^*(B_{h_1}) \cap Z_{h_2}$ belong to different Menger's paths. If $B_{h_2} \neq E^*(B_{h_1}) \cap Z_{h_2}$, then we have:
$\card{\sigma(B_{h_2})} < \card{\sigma\left(E^*(B_{h_1}) \cap Z_{h_2}\right)} \leq \card{\sigma(B_{h_1})}$. Therefore, $B_{h_2} = E^*(B_{h_1}) \cap Z_{h_2}$, so $B_{h_1}$ is closer than $B_{h_2}$ which is contradictory to our assumption that $B_{h_2}$ is a closest dam.
\end{proof}

Moreover, if $B_h$ is a closest dam then its complement $\overline{B}_h = Z_h\backslash B_h$ is also a closest dam. This property will be used to prove the fixed-parameter tractability of our algorithm.

\begin{lemma}
If $B_h$ is a closest dam, then $\overline{B}_h = Z_h\backslash B_h$ is also a closest dam.
\label{le:complementary}
\end{lemma}
\begin{proof}
Suppose that $B_h$ is closest and $\overline{B}_h$ is not: let $\overline{B}_{\alpha}$ denote the closest dam of $\overline{B}_h$, $\alpha < h$.

First, we focus on the dry instance of dam $\overline{B}_{\alpha}$ between levels $\alpha$ and $h$. We prove that no edge/arc $(u,v)$, with the exception of arcs from dam $\overline{B}_{\alpha}$, has one endpoint inside the dry instance of $\overline{B}_{\alpha}$ before level $h$ ({\em i.e.} in the source side of cut $Z_h$) and one outside. We distinguish two cases: 
\begin{itemize}
\item Case 1: If arc $(u,v)$ is such that $u \notin A^*\left(\overline{B}_{\alpha}\right)$ and $v \in A^*\left(\overline{B}_{\alpha}\right)$, then Lemma~\ref{le:enter_dry} brings the contradiction. This argument also holds when $(u,v)$ is undirected.
\item Case 2: If arc $(u,v)$ is such that $u \in A^*\left(\overline{B}_{\alpha}\right)$ is before level $h$ and $v \notin A^*\left(\overline{B}_{\alpha}\right)$, then a Menger's path $Q_j$ leaves the dry instance of $\overline{B}_{\alpha}$ through this arc. However, dams $\overline{B}_{\alpha}$ and $\overline{B}_{h}$ have the same signature, so path $Q_j$ also contains an arc of $\overline{B}_{h}$ placed after arc $(u,v)$. Consequently, there exists an arc $(u',v')$, $u' \notin A^*\left(\overline{B}_{\alpha}\right)$ and $v' \in A^*\left(\overline{B}_{\alpha}\right)$, to make path $Q_j$ go back inside $\mcald(\mcali,\overline{B}_{\alpha})$. This contradicts Case 1. Fig.~\ref{fig:proof_complementary} illustrates the explanations given in Case 2 on a graph $G$ with dams $B_h$, $\overline{B}_h$, $B_{\alpha}$, and $\overline{B}_{\alpha}$. In this example, vertices $u'$ and $v$ are identical.
\end{itemize}
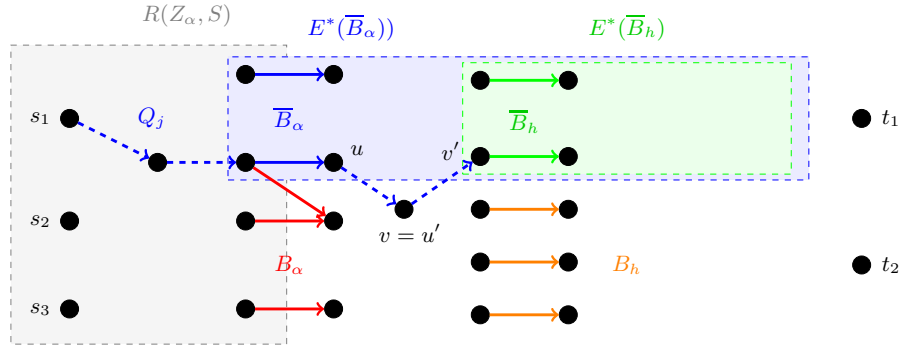
\begin{figure}[h]
\centering
\scalebox{.78}{\begin{tikzpicture}


\draw [dashed, color = gray, fill = white!92!gray] (5.7,0.4) -- (5.7,5.5) -- (1.0,5.5) -- (1.0,0.4) --  (5.7,0.4);
\draw [dashed, color = blue, fill = white!92!blue] (4.7,3.2) -- (14.6,3.2) -- (14.6,5.3) -- (4.7,5.3) -- (4.7,3.2);
\draw [dashed, color = green, fill = white!92!green] (8.7,3.3) -- (14.3,3.3) -- (14.3,5.2) -- (8.7,5.2) -- (8.7,3.3);

\node[draw, circle, minimum height=0.2cm, minimum width=0.2cm, fill=black] (P11) at (2*1,4.25) {};
\node[draw, circle, minimum height=0.2cm, minimum width=0.2cm, fill=black] (P12) at (2*1,2.5) {};
\node[draw, circle, minimum height=0.2cm, minimum width=0.2cm, fill=black] (P13) at (2*1,1) {};

\node[draw, circle, minimum height=0.2cm, minimum width=0.2cm, fill=black] (P21) at (3.5,3.5) {};

\node[draw, circle, minimum height=0.2cm, minimum width=0.2cm, fill=black] (P31) at (5,5) {};
\node[draw, circle, minimum height=0.2cm, minimum width=0.2cm, fill=black] (P32) at (5,3.5) {};
\node[draw, circle, minimum height=0.2cm, minimum width=0.2cm, fill=black] (P33) at (5,2.5) {};
\node[draw, circle, minimum height=0.2cm, minimum width=0.2cm, fill=black] (P34) at (5,1) {};

\node[draw, circle, minimum height=0.2cm, minimum width=0.2cm, fill=black] (P41) at (6.5,5) {};
\node[draw, circle, minimum height=0.2cm, minimum width=0.2cm, fill=black] (P42) at (6.5,3.5) {};
\node[draw, circle, minimum height=0.2cm, minimum width=0.2cm, fill=black] (P43) at (6.5,2.5) {};
\node[draw, circle, minimum height=0.2cm, minimum width=0.2cm, fill=black] (P44) at (6.5,1) {};

\node[draw, circle, minimum height=0.2cm, minimum width=0.2cm, fill=black] (P45) at (7.7,2.7) {};

\node[draw, circle, minimum height=0.2cm, minimum width=0.2cm, fill=black] (P51) at (9,4.9) {};
\node[draw, circle, minimum height=0.2cm, minimum width=0.2cm, fill=black] (P52) at (9,3.6) {};
\node[draw, circle, minimum height=0.2cm, minimum width=0.2cm, fill=black] (P53) at (9,2.7) {};
\node[draw, circle, minimum height=0.2cm, minimum width=0.2cm, fill=black] (P54) at (9,1.8) {};
\node[draw, circle, minimum height=0.2cm, minimum width=0.2cm, fill=black] (P55) at (9,0.9) {};

\node[draw, circle, minimum height=0.2cm, minimum width=0.2cm, fill=black] (P61) at (10.5,4.9) {};
\node[draw, circle, minimum height=0.2cm, minimum width=0.2cm, fill=black] (P62) at (10.5,3.6) {};
\node[draw, circle, minimum height=0.2cm, minimum width=0.2cm, fill=black] (P63) at (10.5,2.7) {};
\node[draw, circle, minimum height=0.2cm, minimum width=0.2cm, fill=black] (P64) at (10.5,1.8) {};
\node[draw, circle, minimum height=0.2cm, minimum width=0.2cm, fill=black] (P65) at (10.5,0.9) {};
%
%

\node[draw, circle, minimum height=0.2cm, minimum width=0.2cm, fill=black] (P91) at (15.5,4.25) {};
\node[draw, circle, minimum height=0.2cm, minimum width=0.2cm, fill=black] (P92) at (15.5,1.75) {};


\draw[->,line width = 1.4pt, color = blue, dashed] (P11) -- (P21);
\draw[->,line width = 1.4pt, color = blue, dashed] (P21) -- (P32);
\draw[->,line width = 1.4pt, color = blue, dashed] (P42) -- (P45);
\draw[->,line width = 1.4pt, color = blue, dashed] (P45) -- (P52);
\draw[->,line width = 1.4pt, color = blue] (P31) -- (P41);
\draw[->,line width = 1.4pt, color = blue] (P32) -- (P42);
\draw[->,line width = 1.4pt, color = red] (P32) -- (P43);
\draw[->,line width = 1.4pt, color = red] (P33) -- (P43);
\draw[->,line width = 1.4pt, color = green] (P51) -- (P61);
\draw[->,line width = 1.4pt, color = green] (P52) -- (P62);
\draw[->,line width = 1.4pt, color = orange] (P53) -- (P63);
\draw[->,line width = 1.4pt, color = orange] (P54) -- (P64);
\draw[->,line width = 1.4pt, color = orange] (P55) -- (P65);
\draw[->,line width = 1.4pt, color = red] (P34) -- (P44);


\node[scale=1.2] at (1.5,4.25) {$s_1$};
\node[scale=1.2] at (1.5,2.5) {$s_2$};
\node[scale=1.2] at (1.5,1) {$s_3$};

\node[scale=1.2] at (6.9,3.7) {$u$};
\node[scale=1.2] at (7.8,2.3) {$v=u'$};
\node[scale=1.2] at (8.5,3.7) {$v'$};

\node[scale=1.2] at (16,4.25) {$t_1$};
\node[scale=1.2] at (16,1.75) {$t_2$};

\node[scale=1.2, color = red] at (5.75,1.75) {$B_{\alpha}$};
\node[scale=1.2, color = orange] at (11.5,1.75) {$B_{h}$};
\node[scale=1.2, color = blue] at (5.75,4.25) {$\overline{B}_{\alpha}$};
\node[scale=1.2, color = blue] at (3.4,4.25) {$Q_j$};
\node[scale=1.2, color = black!20!green] at (9.75,4.2) {$\overline{B}_{h}$};

\node[scale=1.2, color = black!20!green] at (11.5,5.8) {$E^*(\overline{B}_{h})$};
\node[scale=1.2, color = blue] at (6.8,5.8) {$E^*(\overline{B}_{\alpha}))$};
\node[scale=1.2, color = gray] at (4.0,6.0) {$R(Z_{\alpha},S)$};

\end{tikzpicture}}
\caption{Illustration of the contradiction we arose for Case 2 in the proof of Lemma~\ref{le:complementary}.}
\label{fig:proof_complementary}
\end{figure}
Second, we show that any vertex of $V^S(B_h)$ is unreachable from $S$ in $G_{D}\backslash B_{\alpha}$, where $B_{\alpha} = Z_{\alpha} \backslash \overline{B}_{\alpha}$. We suppose that a path $Q$ inside graph $G_{D}\backslash B_{\alpha}$ connects a source $s \in S$ with a vertex $w \in V^S(B_h)$. Path $Q$ necessarily traverses level $\alpha$, so it contains an arc $(u,v)$ of $\overline{B}_{\alpha}$, as the complement dam $B_{\alpha}$ has been removed. As dam $\overline{B}_{\alpha}$ is closer than $\overline{B}_{h}$, vertices of $V^T(\overline{B}_{\alpha})$ form a subset of $A^*\left(\overline{B}_{\alpha}\right)$, otherwise one vertex of $V^T(\overline{B}_{\alpha})$ is reachable in $G_{D}\backslash B_{\alpha}$ and Menger's paths make a vertex in $V^S(\overline{B}_{h})$ be reachable too, which is impossible. For this reason, path $Q$ must contain a vertex $v \in A^*\left(\overline{B}_{\alpha}\right)$. Therefore, it connects a vertex $v$ inside the dry instance of $\overline{B}_{\alpha}$ with vertex $w$ which is outside. As a consequence, there is an edge/arc leaving the dry instance of $\overline{B}_{\alpha}$ on path $Q$, which is a contradiction with our reasoning in Case 2.

Eventually, all vertices of $V^S(B_h)$ are unreachable from $S$ in $G_{D}\backslash B_{\alpha}$, so arcs of $B_h$ belong to the dry instance of $B_{\alpha}$. Conversely, arcs of $\overline{B}_h$ does not belong to $E^*\left(B_{\alpha}\right)$, as the Menger's paths containing arcs of $\overline{B}_{\alpha}$ connect $S$ with $\overline{B}_h$ despite the removal of $B_{\alpha}$. Therefore, $E^*\left(B_{\alpha}\right) \cap Z_h = B_h$. Moreover, $\sigma\left(\overline{B}_{\alpha}\right) = \sigma\left(\overline{B}_{h}\right)$ as $\overline{B}_{\alpha}$ is closer than $\overline{B}_{h}$, so their complement dams have also the same signature: $\sigma\left(B_{\alpha}\right) = \sigma\left(B_h\right)$. Dam $B_{\alpha}$ is thus closer than $B_h$, which is a contradiction because $B_h$ is supposed to be a closest dam.
\end{proof}

Observe that the dry areas of a dam $B_i$ and its complement, $A^*(B_i)$ and $A^*(\overline{B_i})$ respectively, are disjoint because any vertex is reachable from $S$ either in $G\backslash B_i$ or in $G\backslash \overline{B}_i$ or in both of them.

\subsection{A characterization of minimum cuts with dry instances} \label{subsec:keystone}

Theorem~\ref{th:keystone} provides us with the keystone to build our FPT$\langle p \rangle$ algorithm. It combines the concepts of dry instance and closest dam: given a minimum $(S,T)$-cut $X$ and its front dam $B_{i(X)}$, either $X\backslash B_{i(X)} = \emptyset$ and $X=Z_{i(X)}$ or edges in $X\backslash B_{i(X)} \neq \emptyset$ belong to the dry instance of the dam $\overline{B}_{h(X)} = Z_{h(X)}\backslash B_{h(X)}$, where $B_{h(X)}$ is the closest dam of $B_{i(X)}$.

\begin{theorem}
If $X\neq Z_{i(X)}$ is a minimum cut for $\mcali$, $B_{i(X)}$ its front dam, and $B_{h(X)}$ the closest dam of $B_{i(X)}$, then set $X \backslash B_{i(X)}$ is a minimum cut for the dry instance of $\overline{B}_{h(X)} = Z_{h(X)}\backslash B_{h(X)}$, {\em i.e.} $\mcald\left(\mcali, \overline{B}_{h(X)}\right)$.
\label{th:keystone}
\end{theorem}
\begin{proof}
We prove that all edges in $X\backslash B_{i(X)}$ belong to the dry instance of $\overline{B}_{h(X)}$. Let us suppose {\em ad absurdum} that an edge $e=(u,v) \in X\backslash B_{i(X)}$ is reachable from $S$ in graph $G_{D}$ deprived of the dam $\overline{B}_{h(X)}$. Edge $e$ is an arc $(u,v)$ in $G_{D}$.
We denote by $P_e$ a path in $G_{D}$ starting from a source $s$ $\in S$ and terminating with $(u,v)$, deprived of arcs of $\overline{B}_{h(X)}$, $P_e: s\cdots u \cdot v$.
  
Due to the iterative construction of cuts $Z_i$, path $P_e$ necessarily contains one edge $e_{h} = (u_{h},v_{h})$ of level $h(X)$ when it arrives at a level greater than $i(X) \ge h(X)$. We know that this edge $e_{h}$ does not belong to $\overline{B}_{h(X)}$, so $e_{h} \in B_{h(X)}$ and $P_e: s \cdots u_{h}\cdot v_{h}\cdots u\cdot v$. From now on, we focus on the segment of path $P_e$, denoted by $P_e^{(h)}$, which contains all edges between $e_{h}$ and $e$, {\em i.e.} $P_e^{(h)}: u_{h}\cdot v_{h}\cdots u\cdot v$. 
The proof goes in two steps:\\
Step 1: Path $P_e^{(h)}$ contains an arc of $B_{i(X)}$.\\
Step 2: The existence of a path in $G_{D}$ containing both an arc of $B_{i(X)}$ (proven in Step 1) and arc $e$ contradicts the definition of cut $X$.

For Step 1, let us suppose that the path $P_e^{(h)}$ does not contain arcs of $B_{i(X)}$. This means that path $P_e^{(h)}$ passes by the dry instance of $B_{h(X)}$ between levels $h(X)$ and $i(X)$, otherwise it would necessarily contain an arc of $B_{i(X)}$. So, there is an edge/arc of this path, $\tilde{e} = (\tilde{u},\tilde{v})$, where $\tilde{u}$ is in the dry area of $B_{h(X)}$ but not in the dry area of $B_{i(X)}$, and $\tilde{v}$ lies outside the dry area of $B_{h(X)}$. In brief, $\tilde{u} \in A^*(B_{h(X)}) \backslash A^*(B_{i(X)})$ and $\tilde{v} \notin A^*(B_{h(X)})$. Edge $\tilde{e}$ must be an arc $(\tilde{u},\tilde{v})$ in $G_D$ according to Lemma~\ref{le:enter_dry}. Arc $\tilde{e}$ belongs to a Menger's path $Q_{j(\tilde{e})}$ passing through $B_{h(X)}$, {\em i.e.} $Q_{j(\tilde{e})} \in \sigma(B_{h(X)})$. Path $Q_{j(\tilde{e})}$ must traverse an arc of $B_{i(X)}$ as $\sigma(B_{i(X)}) = \sigma(B_{h(X)})$. As a consequence, path $Q_{j(\tilde{e})}$ connects a vertex $\tilde{v}$ outside the dry area of $B_{h(X)}$ with the tail $u_i'$ of an arc $e_i' = (u_i',v_i')$ of $B_{i(X)}$. This is a contradiction, as vertex $u_i'$ is supposed not to be reachable from $S$ in $G_{D}$ deprived of cut $B_{h(X)}$. Path $P_e^{(h)}$ thus contains an arc of $B_{i(X)}$ denoted by $e_i = (u_i,v_i)$.

For Step 2, let $\widehat{e} = \left(\widehat{u},\widehat{v}\right) \neq e_i$ be the first arc of cut $X$ in path $P_e^{(h)}$ which arrives after $e_i$ on this path (Fig.~\ref{fig:proof_keystone}). By this way, we ensure that no edge of $X$ lies on path $P_e^{(h)}$ between vertices $v_i$ and $\widehat{u}$.
Arc $\widehat{e}$ exists as $e$ is a potential candidate to be one.

\begin{figure}[h]
\centering
\scalebox{.78}{\begin{tikzpicture}


\draw [dashed, color = gray, fill = white!92!gray] (5.7,0.4) -- (5.7,5.5) -- (1.0,5.5) -- (1.0,0.4) --  (5.7,0.4);
\draw [dashed, color = red, fill = white!92!red] (4.7,0.5) -- (14.6,0.5) -- (14.6,3.0) -- (4.7,3.0) -- (4.7,0.5);
\draw [dashed, color = green, fill = white!92!green] (7.7,0.6) -- (14.3,0.6) -- (14.3,2.9) -- (7.7,2.9) -- (7.7,0.6);
\draw [dashed, color = blue, fill = white!92!blue] (4.7,5.3) -- (10.6,5.3) -- (10.6,3.2) -- (4.7,3.2) -- (4.7,5.3);

\node[draw, circle, minimum height=0.2cm, minimum width=0.2cm, fill=black] (P11) at (2*1,4.25) {};
\node[draw, circle, minimum height=0.2cm, minimum width=0.2cm, fill=black] (P12) at (2*1,2.5) {};
\node[draw, circle, minimum height=0.2cm, minimum width=0.2cm, fill=black] (P13) at (2*1,1) {};

\node[draw, circle, minimum height=0.2cm, minimum width=0.2cm, fill=black] (P21) at (3.5,2.2) {};

\node[draw, circle, minimum height=0.2cm, minimum width=0.2cm, fill=black] (P31) at (5,5) {};
\node[draw, circle, minimum height=0.2cm, minimum width=0.2cm, fill=black] (P32) at (5,3.5) {};
\node[draw, circle, minimum height=0.2cm, minimum width=0.2cm, fill=black] (P33) at (5,2.5) {};
\node[draw, circle, minimum height=0.2cm, minimum width=0.2cm, fill=black] (P34) at (5,1) {};

\node[draw, circle, minimum height=0.2cm, minimum width=0.2cm, fill=black] (P41) at (6.5,5) {};
\node[draw, circle, minimum height=0.2cm, minimum width=0.2cm, fill=black] (P42) at (6.5,3.5) {};
\node[draw, circle, minimum height=0.2cm, minimum width=0.2cm, fill=black] (P43) at (6.5,2.5) {};
\node[draw, circle, minimum height=0.2cm, minimum width=0.2cm, fill=black] (P44) at (6.5,1) {};

\node[draw, circle, minimum height=0.2cm, minimum width=0.2cm, fill=black] (P53) at (8,2.5) {};
\node[draw, circle, minimum height=0.2cm, minimum width=0.2cm, fill=black] (P54) at (8,1) {};

\node[draw, circle, minimum height=0.2cm, minimum width=0.2cm, fill=black] (P63) at (9.5,2.5) {};
\node[draw, circle, minimum height=0.2cm, minimum width=0.2cm, fill=black] (P65) at (9.5,1) {};
\node[draw, circle, minimum height=0.2cm, minimum width=0.2cm, fill=black] (P66) at (11,2.5) {};

\node[draw, circle, minimum height=0.2cm, minimum width=0.2cm, fill=black] (P71) at (12.5,3.75) {};

\node[draw, circle, minimum height=0.2cm, minimum width=0.2cm, fill=black] (P81) at (14,4.5) {};

\node[draw, circle, minimum height=0.2cm, minimum width=0.2cm, fill=black] (P91) at (15.5,4.25) {};
\node[draw, circle, minimum height=0.2cm, minimum width=0.2cm, fill=black] (P92) at (15.5,1.75) {};


\draw[->,line width = 1.4pt, color = black] (P12) -- (P21);
\draw[->,line width = 1.4pt, color = black] (P21) -- (P33);
\draw[->,line width = 1.4pt, color = blue] (P31) -- (P41);
\draw[->,line width = 1.4pt, color = blue] (P32) -- (P42);
\draw[->,line width = 1.4pt, color = red] (P33) -- (P43);
\draw[->,line width = 1.4pt, color = black] (P43) -- (P54);
\draw[->,line width = 1.4pt, color = green] (P53) -- (P63);
\draw[->,line width = 1.4pt, color = green] (P54) -- (P65);
\draw[->,line width = 1.4pt, color = black] (P65) -- (P66);
\draw[->,line width = 1.4pt, color = red] (P34) -- (P44);

\draw[->,line width = 1.4pt, color = black] (P66) -- (P71);
\draw[->,line width = 1.4pt, color = black] (P71) -- (P81);


\node[scale=1.2] at (1.5,4.25) {$s_1$};
\node[scale=1.2] at (1.5,2.5) {$s_2$};
\node[scale=1.2] at (1.5,1) {$s_3$};

\node[scale=1.2] at (12,3.9) {$\widehat{u}$};
\node[scale=1.2] at (14.5,4.6) {$\widehat{v}$};
\node[scale=1.2] at (4.4,2.7) {$u_h$};
\node[scale=1.2] at (7.0,2.6) {$v_h$};
\node[scale=1.2] at (7.4,1.0) {$u_i$};
\node[scale=1.2] at (10.0,1.0) {$v_i$};
\node[scale=1.2] at (8.75,1.2) {$e_i$};
\node[scale=1.2] at (13.1,4.4) {$\widehat{e}$};

\node[scale=1.2] at (16,4.25) {$t_1$};
\node[scale=1.2] at (16,1.75) {$t_2$};

\node[scale=1.2, color = red] at (5.75,1.75) {$B_{h(X)}$};
\node[scale=1.2, color = blue] at (5.75,4.25) {$\overline{B}_{h(X)}$};
\node[scale=1.2, color = black!20!green] at (8.75,1.75) {$B_{i(X)}$};

\node[scale=1.2, color = red] at (6.7,0.1) {$E^*(B_{h(X)})$};
\node[scale=1.2, color = black!20!green] at (13.0,1.0) {$E^*(B_{i(X)})$};
\node[scale=1.2, color = blue] at (8.2,5.8) {$E^*(\overline{B}_{h(X)}))$};
\node[scale=1.2, color = gray] at (4.0,6.0) {$R(Z_{h(X)},S)$};

\end{tikzpicture}}
\caption{Illustration of the segment of path $P_e$ between $S$ and $\widehat{e}$, traversing dams $B_h$ and $B_i$.}
\label{fig:proof_keystone}
\end{figure}
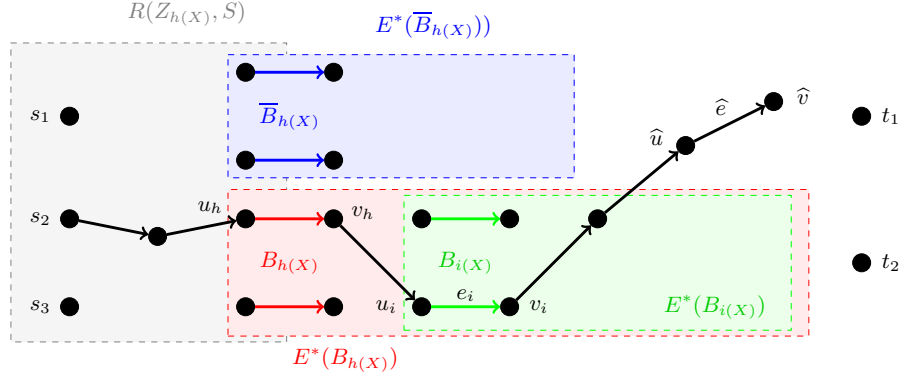
As vertex $v_i$ is the head of $e_i$ in graph $G_{D}$, $v_i \in R(X,T)$ according to Lemma~\ref{le:menger}. For the same reason, vertex $\widehat{u} \in R(X,S)$ as it is the tail of arc $\widehat{e} \in X$. We know that path $P_e^{(h)}$ in $G_{D}$ connects these two vertices and there is no arc of $X$ on the segment of $P_e^{(h)}$ connecting them. Let us go back now to the initial undirected graph $G$. In graph $G\backslash X$, vertices $v_i \in R(X,T)$ and $\widehat{u} \in R(X,S)$ are connected whereas they must be separated by $X$. The presence of an edge $(u,v)$ of $X$ outside $E^*(\overline{B}_{h(X)})$ yields a contradiction.

In summary, all edges in $X\backslash B_{i(X)}$ belong to the dry instance of $\overline{B}_{h(X)}$. They necessarily form a cut in instance $\mcald\left(\mcali, \overline{B}_{h(X)}\right)$, otherwise $X$ would not separate $S$ and $T$ in $\mcali$. The $p-\card{B_{h(X)}}$ Menger's paths in signature $\sigma(\overline{B}_{h(X)})$ are edge-disjoint inside $\mcald\left(\mcali, \overline{B}_{h(X)}\right)$, so the minimum cut size of this instance is greater than $p-\card{B_{h(X)}}$. As $X\backslash B_{i(X)}$ contains $p-\card{B_{i(X)}}=p-\card{B_{h(X)}}$ edges, we conclude that it is a minimum cut of $\mcald\left(\mcali, \overline{B}_{h(X)}\right)$.
\end{proof}

Therefore, any minimum $(S,T)$-cut $X$, which is not a minimum drainage cut $Z_{i(X)}$ itself, can be partitioned into two sets, $B_{i(X)}$ and $X\backslash B_{i(X)}$, such that:
\begin{itemize}
\item $B_{i(X)}$ is a minimum cut of instance $\mcald\left(\mcali, B_{h(X)}\right)$ and a dam of $\mcali$,
\item $X\backslash B_{i(X)}$ is a minimum cut of instance $\mcald\left(\mcali, \overline{B}_{h(X)}\right)$ and all its edges belong to the target side of $Z_{i(X)}$, $E\left[R(Z_{i(X)},T)\right]$.
\end{itemize}

Conversely, given a closest dam $B_h$ and its complement $\overline{B}_h = Z_h\backslash B_h$, the union $B_i \cup X_{\overline{B}_h}$, where the closest dam of $B_i$ is $B_h$ and $X_{\overline{B}_h}$ is a minimum cut of $\mcald\left(\mcali, \overline{B}_{h}\right)$, separates $S$ from $T$.
\begin{theorem}
Let $B_h$ be a closest dam of $\mcalz(\mcali)$ and $\overline{B}_h = Z_h\backslash B_h$. Let $B_i$ be a dam such that $B_h$ is closer than $B_i$ and $X_{\overline{B}_h}$ a minimum cut of $\mcald\left(\mcali, \overline{B}_{h}\right)$. Then, $B_i \cup X_{\overline{B}_h}$  is a minimum $(S,T)$-cut for instance~$\mcali$.
\label{th:converse_keystone}
\end{theorem}
\begin{proof}
As dam $B_h$ is closer than $B_i$, the edges of $B_i$ form a minimum cut of $\mcald\left(\mcali,B_h\right)$. Indeed, they are the edges of level $i$ inside $\mcald\left(\mcali,B_h\right)$, so they separate $S^*(B_h)$ from $T^*(B_h)$. Moreover, we know there is a set of $\card{\sigma(B_h)}$ edge-disjoint paths from $S^*(B_h)$ to $T^*(B_h)$ in $\mcald\left(\mcali,B_h\right)$. As $\card{\sigma(B_h)} = \card{\sigma(B_i)} = \card{B_i}$, set $B_i$ is a minimum cut of instance $\mcald\left(\mcali,B_h\right)$.

We suppose that there is an open $(S,T)$-path $Q$ in undirected graph $G$ deprived of edges $B_i \cup X_{\overline{B}_h}$. Path $Q$ cannot avoid level $h$ of the drainage and passes through one edge of $Z_h$. As $B_h \cup \overline{B}_h = Z_h$, some edges of path $Q$ belong either to the dry instance of $B_h$ or to the dry instance of $\overline{B}_h$, or to both of them.

First, from Lemma~\ref{le:enter_dry} we know that no edge of graph $G$ has one endpoint in the dry area $A^*(B_h)$ of $B_h$ and the another one in the dry area $A^*(\overline{B}_h)$ of $\overline{B}_h$. 

Second, we show that the existence of path $Q$ yields a contradiction with the definition of the dry instance. As sets $A^*(B_h)$ and $A^*(\overline{B}_h)$ cannot be connected by an edge of $G$, path $Q$ ``traverses'' completely at least one of the dry instances $\mcald\left(\mcali,B_h\right)$ or $\mcald\left(\mcali,\overline{B}_h\right)$, with no loss of generality we say $\mcald\left(\mcali,B_h\right)$. In other words, a segment of $Q$ connects $S^*(B_h)$ and $T^*(B_h)$. This is not possible because dam $B_i$ separates these two sets of vertices. With the dry instance $\mcald\left(\mcali,\overline{B}_h\right)$, we obtain the same contradiction as $X_{\overline{B}_h}$ separates $S^*(\overline{B}_h)$ and $T^*(\overline{B}_h)$. Therefore, $B_i \cup X_{\overline{B}_h}$ separates $S$ from $T$ with $p$ edges, as $\card{X_{\overline{B}_h}} = \card{\sigma(\overline{B}_h)} = p-\card{B_i}$.
\end{proof}

We now prove a stronger result for set $X\backslash B_{i(X)}$. In fact, edges of set $X\backslash B_{i(X)}$ lie in the target side of level $i(X)-h(X)+1$ in the drainage of instance $\mcald\left(\mcali,\overline{B}_{h(X)}\right)$. This statement is formulated in the theorem below. 

\begin{theorem}
Let $X$ be a minimum $(S,T)$-cut of $G$ and let $(Z_1',\ldots,Z_{k'}')$ be the drainage of instance $\mcald\left(\mcali,\overline{B}_{h(X)}\right)$. Then, set $Z_{i(X)-h(X)+1}'$ is equal to $\overline{B}_{i(X)} = Z_{i(X)}\backslash B_{i(X)}$ and edges $X\backslash B_{i(X)}$ belong to the target side of $Z_{i(X)-h(X)+1}'$ inside instance $\mcald\left(\mcali,\overline{B}_{h(X)}\right)$.
\label{th:recursive}
\end{theorem}
\begin{proof}
According to Theorem~\ref{th:keystone}, we know that edges of $X\backslash B_{i(X)}$ belong to the dry instance of $\overline{B}_{h(X)}$, {\em i.e.} $E^*(\overline{B}_{h(X)})$. Moreover, they are also in the target side of $Z_{i(X)}$, as $B_{i(X)} \subsetneq Z_{i(X)}$ is the front dam of $X$. We denote by $\overline{B}_{i(X)}$ the complement of $B_{i(X)}$ in $Z_{i(X)}$, $\overline{B}_{i(X)} = Z_{i(X)}\backslash B_{i(X)}$. We want to prove that dam $\overline{B}_{i(X)}$ is the minimum drainage cut of level ${i(X)} - {h(X)} + 1$ in instance $\mcald\left(\mcali,\overline{B}_{h(X)}\right)$. For this purpose, we first prove that $\overline{B}_{h(X)}$ is closer than $\overline{B}_{i(X)}$. 

Dam $\overline{B}_{i(X)}$ has the same signature as $\overline{B}_{h(X)}$ because their respective complement fulfil $\sigma(B_{i(X)}) = \sigma(B_{h(X)})$. Moreover, we prove that arcs of $\overline{B}_{i(X)}$ are in the dry instance of $\overline{B}_{h(X)}$. Suppose an arc $e_i = (u_i,v_i)$ of $\overline{B}_{i(X)}$ does not belong to $E^*\left(\overline{B}_{h(X)}\right)$. Let $Q_j$ be the Menger's path  containing arc $e_i$. Then, no arc of $Q_j$ after $e_i$ is inside instance $\mcald\left(\mcali,\overline{B}_{h(X)}\right)$. This is a contradiction as path $Q_j$ must contain an edge of cut $X$. Indeed, path $Q_j$ contains neither an arc of $B_{i(X)}$ as $\sigma(B_{i(X)}) \cap \sigma(\overline{B}_{i(X)}) = \emptyset$ nor an arc of $X\backslash B_{i(X)}$ as all its arcs after level ${i(X)}$ do not belong to $\mcald\left(\mcali,\overline{B}_{h(X)}\right)$ (contradiction with Theorem~\ref{th:keystone}). In summary, dam $\overline{B}_{h(X)}$ is closer than $\overline{B}_{i(X)}$.

All dams $\overline{B}_{\ell}$ such that $\sigma(\overline{B}_{h(X)}) = \sigma(\overline{B}_{\ell})$ and ${h(X)} < \ell < {i(X)}$ have a common closest dam: $\overline{B}_{h(X)}$. Indeed, if it is not the case for a dam $\overline{B}_{\ell}$, there is an arc $(u_{\ell},v_{\ell}) \in \overline{B}_{\ell}$ where $u_{\ell} \notin A^*\left(\overline{B}_{h(X)}\right)$. As a consequence, dam $\overline{B}_{i(X)}$ is not contained in $E^*(\overline{B}_{h(X)})$, as arc $(u_{\ell},v_{\ell})$ belongs to a Menger's path containing an arc of $\overline{B}_{i(X)}$.

\begin{figure}[t]
\centering
\scalebox{.78}{\begin{tikzpicture}


\draw [dashed, color = gray] (16.8,0.2) -- (16.8,5.7) -- (1.0,5.7) -- (1.0,0.2) --  (16.8,0.2);

\draw [dashed, color = black, fill = white!90!gray] (1.3,0.5) -- (2.0,0.5) -- (2.0,5.3) -- (1.3,5.3) --  (1.3,0.5);
\draw [dashed, color = black, fill = white!90!gray] (15.8,0.5) -- (16.5,0.5) -- (16.5,5.3) -- (15.8,5.3) --  (15.8,0.5);

\draw [dashed, color = red, fill = white!92!red] (5.1,0.5) -- (13.6,0.5) -- (13.6,3.0) -- (5.1,3.0) -- (5.1,0.5);
\draw [dashed, color = red, fill = white!92!red] (3.7,0.5) -- (4.3,0.5) -- (4.3,3.0) --(3.7,3.0) -- (3.7,0.5);
\draw [dashed, color = red, fill = white!92!red] (14.2,0.5) -- (14.9,0.5) -- (14.9,3.0) -- (14.2,3.0) -- (14.2,0.5);


\draw [dashed, color = blue, fill = white!92!blue] (3.7,5.3) -- (4.3,5.3) -- (4.3,3.2) --(3.7,3.2) -- (3.7,5.3);
\draw [dashed, color = blue, fill = white!92!blue] (5.1,5.3) -- (12.6,5.3) -- (12.6,4.2) -- (10.6,4.2) -- (10.6,3.2) --(5.1,3.2) -- (5.1,5.3);
\draw [dashed, color = blue, fill = white!92!blue] (11.2,3.2) -- (12.0,3.2) -- (12.0,3.9) -- (14.0,3.9) -- (14.0,5.3) --(13.2,5.3) -- (13.2,4.0) -- (11.2,4.0) -- (11.2,3.2);

\draw [dashed, color = yellow, fill = white!92!yellow] (8.1,3.3) -- (10.4,3.3) -- (10.4,4.3) -- (12.4,4.3) -- (12.4,5.2) -- (8.1,5.2) -- (8.1,3.3);

\node[draw, circle, minimum height=0.2cm, minimum width=0.2cm, fill=black] (P31) at (4,5) {};
\node[draw, circle, minimum height=0.2cm, minimum width=0.2cm, fill=black] (P32) at (4,3.5) {};
\node[draw, circle, minimum height=0.2cm, minimum width=0.2cm, fill=black] (P33) at (4,2.5) {};
\node[draw, circle, minimum height=0.2cm, minimum width=0.2cm, fill=black] (P34) at (4,1) {};

\node[draw, circle, minimum height=0.2cm, minimum width=0.2cm, fill=black] (P41) at (5.5,5) {};
\node[draw, circle, minimum height=0.2cm, minimum width=0.2cm, fill=black] (P42) at (5.5,3.5) {};
\node[draw, circle, minimum height=0.2cm, minimum width=0.2cm, fill=black] (P43) at (5.5,2.5) {};
\node[draw, circle, minimum height=0.2cm, minimum width=0.2cm, fill=black] (P44) at (5.5,1) {};

\node[draw, circle, minimum height=0.2cm, minimum width=0.2cm, fill=black] (P51) at (7,5) {};
\node[draw, circle, minimum height=0.2cm, minimum width=0.2cm, fill=black] (P52) at (7,3.5) {};
\node[draw, circle, minimum height=0.2cm, minimum width=0.2cm, fill=black] (P53) at (7,2.5) {};
\node[draw, circle, minimum height=0.2cm, minimum width=0.2cm, fill=black] (P54) at (7,1) {};

\node[draw, circle, minimum height=0.2cm, minimum width=0.2cm, fill=black] (P61) at (8.5,5) {};
\node[draw, circle, minimum height=0.2cm, minimum width=0.2cm, fill=black] (P62) at (8.5,3.5) {};
\node[draw, circle, minimum height=0.2cm, minimum width=0.2cm, fill=black] (P63) at (8.5,2.5) {};
\node[draw, circle, minimum height=0.2cm, minimum width=0.2cm, fill=black] (P64) at (8.5,1) {};

\node[draw, circle, minimum height=0.2cm, minimum width=0.2cm, fill=black] (P71) at (10.1,3.7) {};
\node[draw, circle, minimum height=0.2cm, minimum width=0.2cm, fill=black] (P72) at (12.1,4.75) {};
\node[draw, circle, minimum height=0.2cm, minimum width=0.2cm, fill=black] (P73) at (13.1,2.5) {};
\node[draw, circle, minimum height=0.2cm, minimum width=0.2cm, fill=black] (P74) at (13.1,1.0) {};

\node[draw, circle, minimum height=0.2cm, minimum width=0.2cm, fill=black] (P81) at (11.6,3.7) {};
\node[draw, circle, minimum height=0.2cm, minimum width=0.2cm, fill=black] (P82) at (13.6,4.75) {};
\node[draw, circle, minimum height=0.2cm, minimum width=0.2cm, fill=black] (P83) at (14.6,2.5) {};
\node[draw, circle, minimum height=0.2cm, minimum width=0.2cm, fill=black] (P84) at (14.6,1.0) {};


\draw[->,line width = 1.4pt, color = blue] (P31) -- (P41);
\draw[->,line width = 1.4pt, color = blue] (P32) -- (P42);
\draw[->,line width = 1.4pt, color = red] (P33) -- (P43);
\draw[->,line width = 1.4pt, color = black!20!yellow] (P51) -- (P61);
\draw[->,line width = 1.4pt, color = black!20!yellow] (P52) -- (P62);
\draw[->,line width = 1.4pt, color = green] (P53) -- (P63);
\draw[->,line width = 1.4pt, color = green] (P54) -- (P64);
\draw[->,line width = 1.4pt, color = red] (P34) -- (P44);
\draw[->,line width = 1.4pt, color = black] (P71) -- (P81);
\draw[->,line width = 1.4pt, color = black] (P72) -- (P82);
\draw[->,line width = 1.4pt, color = black] (P73) -- (P83);
\draw[->,line width = 1.4pt, color = black] (P74) -- (P84);


\node[scale=1.2] at (1.65,3.0) {$S$};
\node[scale=1.2, color=blue] at (4.0,6.0) {$S^*(\overline{B}_h)$};
\node[scale=1.2, color=red] at (4.0,-0.3) {$S^*(B_h)$};
\node[scale=1.2, color=blue] at (9.0,6.0) {$A^*(\overline{B}_h)$};
\node[scale=1.2, color=red] at (9.0,-0.3) {$A^*(B_h)$};
\node[scale=1.2, color = black!20!yellow] at (9.5,4.5) {$A^*(\overline{B}_i)$};

\draw[->,line width = 1.4pt, color = black] (4.75,6.7) -- (14.3,6.7);
\draw[line width = 1.4pt, color = black] (4.75,6.5) -- (4.75,6.9);
\draw[line width = 1.4pt, color = black] (7.75,6.5) -- (7.75,6.9);
\node[scale=1.2, text width = 2.0cm] at (15.5,6.7) {Levels of $\mcald(\mcali,\overline{B}_h)$};
\node at (7.75,7.1) {$i-h+1$};
\node at (4.75,7.1) {$1$};

\node[scale=1.2] at (16.15,3.0) {$T$};
\node[scale=1.2, color=blue] at (13.6,6.0) {$T^*(\overline{B}_h)$};
\node[scale=1.2, color=red] at (14.6,-0.3) {$T^*(B_h)$};

\node[scale=1.2, color = red] at (4.75,1.75) {$B_h$};
\node[scale=1.2, color = blue] at (4.75,4.25) {$\overline{B}_h$};
\node[scale=1.2, color = black!20!green] at (7.75,1.75) {$B_i$};
\node[scale=1.2, color = black!20!yellow] at (7.75,4.25) {$\overline{B}_i$};


\end{tikzpicture}}
\caption{Illustration of Theorem~\ref{th:recursive}: for any minimum cut with the front dam $B_i$, the tails of arcs in $X\backslash B_{i}$ belong to the yellow zone.}
\label{fig:equation}
\end{figure}
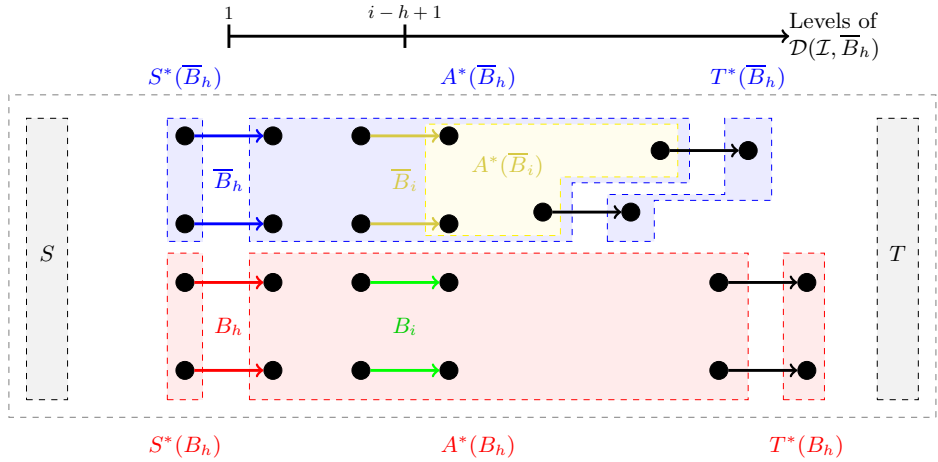

We focus now on the drainage of instance $\mcald\left(\mcali,\overline{B}_{h(X)}\right)$. The minimum drainage cut of level one in this instance is $\overline{B}_{h(X)}$ itself, as it is the minimum closest $\left(S^*(\overline{B}_{h(X)}),T^*(\overline{B}_{h(X)})\right)$-cut. 

Let $\overline{B}_{{h(X)}+1}$ be the dam of level ${h(X)}+1$ in $\mcali$ fulfilling $\sigma(\overline{B}_{h(X)}) = \sigma(\overline{B}_{{h(X)}+1})$ and $B_{{h(X)}+1}$ its complement. We prove that $\overline{B}_{{h(X)}+1}$ is the minimum drainage cut of level two in $\mcald\left(\mcali,\overline{B}_{h(X)}\right)$, {\em i.e.} it is the minimum closest cut between $V^T(\overline{B}_{h(X)})$ and $T^*(\overline{B}_{h(X)})$ in graph $G^*(\overline{B}_{h(X)})$ deprived of $R\left(\overline{B}_{h(X)},S^*(\overline{B}_{h(X)})\right)$. Suppose that there is another minimum closest cut $Z_{\overline{B}_{h(X)}} \neq \overline{B}_{{h(X)}+1}$. Set $X_{\overline{B}_{h(X)}} = Z_{\overline{B}_{h(X)}} \cup B_{{h(X)}+1}$ is a minimum $(S,T)$-cut of instance $\mcali$ according to Theorem~\ref{th:converse_keystone}. As its edges belong to the target side of $Z_{h(X)}$, it is a minimum $(S_{{h(X)}+1},T)$-cut. Based on the definition of $Z_{\overline{B}_{h(X)}}$, the reachable set of $X_{\overline{B}_{h(X)}}$ is necessarily included into the reachable set of $Z_{{h(X)}+1} = B_{{h(X)}+1} \cup \overline{B}_{{h(X)}+1}$ in graph $G\backslash R(Z_{h(X)},S)$. This is a contradiction to the construction of the drainage, as $Z_{{h(X)}+1}$ is the unique minimum closest $(S_{{h(X)}+1},T)$-cut in graph $G\backslash R(Z_{h(X)},S)$. Consequently, $\overline{B}_{{h(X)}+1}$ is the minimum drainage cut of level two inside $\mcald\left(\mcali,\overline{B}_{h(X)}\right)$.

We can iterate these arguments on dams $\overline{B}_{{h(X)}+2}$, $\overline{B}_{{h(X)}+3}$, etc. For example, dam $\overline{B}_{{h(X)}+2}$, $\sigma(\overline{B}_{{h(X)}+2}) = \sigma(\overline{B}_{h(X)})$, is the minimum closest $\left(V^T(\overline{B}_{{h(X)}+1}),T\right)$-cut when graph $G^*(\overline{B}_{h(X)})$ is deprived of $R\left(\overline{B}_{{h(X)}+1},S^*(\overline{B}_{h(X)})\right)$. Otherwise, it would imply that $Z_{{h(X)}+2}$ is not the minimum closest $(S_{{h(X)}+2},T)$-cut in graph $G\backslash R(Z_{{h(X)}+1},S)$, which contradicts the construction of the drainage. Eventually, dam $\overline{B}_{{h(X)}+2}$ is the minimum drainage cut of level three inside $\mcald\left(\mcali,\overline{B}_{h(X)}\right)$, dam $\overline{B}_{{h(X)}+3}$ of level four, etc. Then, dam $\overline{B}_{i(X)}$ is the minimum drainage cut of level ${i(X)}-{h(X)}+1$ in instance $\mcald\left(\mcali,\overline{B}_{h(X)}\right)$.

Coming back to Theorem~\ref{th:keystone}, the edges of $X\backslash B_{i(X)}$ belong to both $E^*(\overline{B}_{h(X)})$ and $E\left[R(Z_{i(X)},T)\right]$. So, they are in the target side of dam $\overline{B}_{i(X)}$ inside instance $\mcald\left(\mcali,\overline{B}_{h(X)}\right)$.
\end{proof}

\subsection{Description of the algorithm} \label{subsec:algorithm}

Our algorithm starts by computing the drainage $\mcalz(\mcali)$ and the Menger's paths of input instance $\mcali$.
For all dams $B_i$, it counts the minimum cuts of size $p$ in $\mcali$ which admit the front dam $B_i$. If $B_i \neq Z_i$, it does this recursively by counting the minimum cuts in instance $\mcald\left(\mcali,\overline{B}_h\right)$ which only contains edges from the target side of the internal level $i-h+1$ of $\mcald\left(\mcali,\overline{B}_h\right)$, where $B_h$ is the closest dam of $B_i$. The minimum cut size in $\mcald(\mcali,\overline{B}_h)$ is at most $p-1$.

We denote by $C_0(\mcali)=C(\mcali)$ the total number of minimum $(S,T)$-cuts of instance $\mcali$. We define $C_{\ell}(\mcali)$ as the number of minimum cuts of instance $\mcali$ which are composed of edges from $E\left[R(Z_{\ell},T)\right]$ only. For example, $C_{2}(\mcali)$ gives the number of minimum $(S,T)$-cuts in instance $\mcali$ without edges of $Z_1 \cup Z_2$. Value $C_{\ell}(\mcali)$, $0\le \ell \le k-1$, can be written:
\begin{equation}
C_{\ell}(\mcali) = k-\ell + \sum_{\substack{\mbox{\scriptsize{Closest}}\\ {\mbox{\scriptsize{dam }}} B_h\subsetneq Z_h}} \sum_{\substack{i ~:~ i > \ell, \\\exists B_i : B_h \\ \mbox{\scriptsize{closer than }} B_{i}}}  C_{i-h+1}\left(\mcald\left(\mcali,\overline{B}_h\right)\right).
\label{eq:counting}
\end{equation}

The first $k-\ell$ cuts are the minimum drainage cuts of $\mcali$ with level greater than $\ell$, {\em i.e.} cuts $Z_{\ell+1},\ldots,Z_k$. The second term counts cuts taking edges only from $E\left[R(Z_{\ell},T)\right]$ and admitting a front dam $B_{i(X)} \neq Z_{i(X)}$. Theorems~\ref{th:keystone} and~\ref{th:recursive} guarantee that any of these minimum $(S,T)$-cuts is counted at least once. Indeed, for any front dam $B_i$ and its closest dam $B_h$, we compute the number of cuts in instance $\mcald\left(\mcali, \overline{B}_{h}\right)$ such that all their edges belong to the target side of $\overline{B}_i$, which is the internal level $i-h+1$ in $\mcald\left(\mcali, \overline{B}_{h}\right)$. In the event that the drainage of $\mcald\left(\mcali, \overline{B}_{h}\right)$ has less than $i-h+1$ levels, then $C_{i-h+1}\left(\mcald\left(\mcali,\overline{B}_h\right)\right) = 0$, as it means no minimum cut of $\mcali$ has the front dam $B_i$.

Conversely, the unicity of a closest dam ensures us that each minimum cut is counted exactly once. A minimum $(S,T)$-cut $X\neq Z_{i(X)}$ has a unique front dam $B_{i(X)}$ and the closest dam $B_{h(X)}$ of $B_{i(X)}$ is unique (Lemma~\ref{le:unique_closest}). Finally, Theorem~\ref{th:converse_keystone} secures that all cuts counted with Eq.~\eqref{eq:counting} are minimum $(S,T)$-cuts.

Value $C_0(\mcali)$ is computed thanks to recursive calls on multiple instances $\mcald\left(\mcali, \overline{B}_{h}\right)$. From now on, we distinguish the input instance $\mcali$ (for which we want to compute $C_0(\mcali)$) with other instances (denoted by $\mcalj$ later on) of the recursive tree. The base cases of the recursion, {\em i.e.} the leaves of the recursive tree, are the computation of values $C_{\ell}(\mcalj)$ either in instances $\mcalj$ where the minimum cut size is one or in instances where no minimum cut admits a front dam $B_i \neq Z_i$, $i > \ell$. In both cases, the only minimum cuts of $\mcalj$ are its minimum drainage cuts.
Each recursive call of the algorithm makes the minimum cut size decrease: for example, if the minimum cut size of $\mcalj$ is $q$, then it is $\card{\overline{B}_h} < q$ for an instance $\mcald\left(\mcalj, \overline{B}_{h}\right)$. Therefore, the recursive tree is not deeper than $p-1$.


Fig.~\ref{fig:recurrence} illustrates the recursive scheme of our algorithm with a tree describing the relationship between the instances. Three instances $\mcali$, $\mcalj$, and $\mcalj'$ of the recursive tree are represented. For example, instance $\mcalj = \mcald\left(\mcali,\overline{B}_h\right)$ is the son of instance $\mcali$ in the tree as it is one of its dry instances. 

In parallel, another graph (black dashed arcs in Fig.~\ref{fig:recurrence}) contains arcs with endpoints $C_{\ell}\left(\mcalj\right)$. An arc connects two ``compartments'' $C_{\ell}\left(\mcalj\right)$ and $C_{\ell'}\left(\mcalj'\right)$ when the computation of $C_{\ell}\left(\mcalj\right)$ depends on $C_{\ell'}\left(\mcalj'\right)$. The minimum $(S,T)$-cut size of $\mcalj'$ is smaller than the one of $\mcalj$. This is why the graph made of arcs between compartments is a DAG. 

\begin{figure}[h]
\centering
\scalebox{.62}{\begin{tikzpicture}


\draw [color = black, fill = white] (6.0,10.0) -- (6.0,12.0) -- (13.0,12.0) -- (13.0,10.0) --  (6.0,10.0);
\draw [color = black, fill = white] (3.0,6.0) -- (3.0,8.0) -- (10.0,8.0) -- (10.0,6.0) -- (3.0,6.0);
\draw [color = black, fill = white] (1.0,2.0) -- (1.0,4.0) -- (8.0,4.0) -- (8.0,2.0) -- (1.0,2.0);

\node[draw, minimum height=0.2cm, minimum width=0.2cm] (P11) at (7.0,10.5) {87};
\node[draw, minimum height=0.2cm, minimum width=0.2cm] (P12) at (8.0,10.5) {81};
\node[scale=1.2] (P13) at (9.0,10.4) {$\cdots$};
\node[draw, minimum height=0.2cm, minimum width=0.2cm] (P14) at (10.0,10.5) {67};
\node[scale=1.2] (P15) at (11.0,10.4) {$\cdots$};
\node[draw, minimum height=0.2cm, minimum width=0.2cm] (P16) at (12.0,10.5) {2};

\node[draw, minimum height=0.2cm, minimum width=0.2cm] (P21) at (4.0,6.5) {16};
\node[draw, minimum height=0.2cm, minimum width=0.2cm] (P22) at (5.0,6.5) {15};
\node[scale=1.2] (P23) at (6.0,6.4) {$\cdots$};
\node[draw, minimum height=0.2cm, minimum width=0.2cm] (P24) at (7.0,6.5) {11};
\node[scale=1.2] (P25) at (8.0,6.4) {$\cdots$};
\node[draw, minimum height=0.2cm, minimum width=0.2cm] (P26) at (9.0,6.5) {1};

\node[draw, minimum height=0.2cm, minimum width=0.2cm] (P31) at (2.0,2.5) {4};
\node[draw, minimum height=0.2cm, minimum width=0.2cm] (P32) at (3.0,2.5) {2};
\node[scale=1.2] (P33) at (4.0,2.4) {$\cdots$};
\node[draw, minimum height=0.2cm, minimum width=0.2cm] (P34) at (5.0,2.5) {1};
\node[scale=1.2] (P35) at (6.0,2.4) {$\cdots$};
\node[draw, minimum height=0.2cm, minimum width=0.2cm] (P36) at (7.0,2.5) {0};


\node[scale=0.8] at (6.9,11.0) {$C_0\left(\mcali\right)$};
\node[scale=0.8] at (8.1,11.0) {$C_1\left(\mcali\right)$};
\node[scale=0.8] at (10.0,11.0) {$C_{\ell}\left(\mcali\right)$};
\node[scale=0.8] at (12.0,11.0) {$C_{k-1}\left(\mcali\right)$};

\node[scale=0.8] (E21) at (3.9,7.1) {$C_0\left(\mcalj\right)$};
\node[scale=0.8] (E22) at (5.1,7.1) {$C_1\left(\mcalj\right)$};
\node[scale=0.8] (E23) at (7.0,7.1) {$C_{\ell}\left(\mcalj\right)$};
\node[scale=0.8] (E24) at (9.0,7.1) {$C_{k-1}\left(\mcalj\right)$};

\node[scale=0.8] (E31) at (1.9,3.1) {$C_0\left(\mcalj'\right)$};
\node[scale=0.8] (E32) at (3.1,3.1) {$C_1\left(\mcalj'\right)$};
\node[scale=0.8] (E33) at (5.0,3.1) {$C_{\ell}\left(\mcalj'\right)$};
\node[scale=0.8] (E34) at (7.0,3.1) {$C_{k-1}\left(\mcalj'\right)$};

\node[scale=1.1] at (4.8,11.2) {Instance $\mcali$};
\node[scale=1.1] at (1.4,7.5) {Instance};
\node[scale=1.1] at (1.4,6.9) {$\mcalj = \mcald\left(\mcali,\overline{B}_h\right)$};
\node[scale=1.1] at (-0.8,3.5) {Instance};
\node[scale=1.1] at (-0.8,2.9) {$\mcalj' = \mcald\left(\mcalj,\overline{B}_{h'}\right)$};

\node[scale=1.1, color = blue] at (14.8,10.6) {Depth $0$};
\node[scale=1.1, color = blue] at (14.8,6.6) {Depth $1$};
\node[scale=1.1, color = blue] at (14.8,2.6) {Depth $2$};

\draw[line width = 2pt, color=blue, dashed] (13.5,11.0)--(15.8,11.0);
\draw[line width = 2pt, color=blue, dashed] (10.5,7.0)--(15.8,7.0);
\draw[line width = 2pt, color=blue, dashed] (8.5,3.0)--(15.5,3.0);
\draw[->,>=latex,line width = 2pt, color=blue] (16.2,12.0)--(16.2,2.0);

\draw[->,>=latex,rounded corners=5pt,line width = 1.4pt] (6.0,10.3) -| (3.5,8.0);
\draw[->,>=latex,rounded corners=5pt,line width = 1.4pt] (3.0,6.3) -| (1.5,4.0);

\draw[->,>=latex,line width = 1.4pt] (11.3,10.0) -- (11.3,8.0);
\draw[->,>=latex,line width = 1.4pt] (11.7,10.0) -- (11.7,8.0);
\draw[->,>=latex,line width = 1.4pt] (12.1,10.0) -- (12.1,8.0);
\draw[->,>=latex,rounded corners=5pt,line width = 1.4pt] (13.0,10.3) -| (14.5,8.0);
\node[scale=1.8] at (13.3,9.0) {$\cdots$};

\draw[->,>=latex,line width = 1.4pt] (8.6,6.0) -- (8.6,4.0);
\draw[->,>=latex,line width = 1.4pt] (9.0,6.0) -- (9.0,4.0);
\draw[->,>=latex,line width = 1.4pt] (9.4,6.0) -- (9.4,4.0);
\draw[->,>=latex,rounded corners=5pt,line width = 1.4pt] (10.0,6.3) -| (11.7,4.0);
\node[scale=1.8] at (10.5,5.0) {$\cdots$};

\draw[->, dashed] (P11) -- (E21);
\draw[->, dashed] (P11) -- (E22);
\draw[->, dashed] (P11) -- (E23);
\draw[->, dashed] (P11) -- (E24);
\draw[->, dashed] (P11) -- (11.3,7.4);
\draw[->, dashed] (P11) -- (11.7,7.4);
\draw[->, dashed] (P11) -- (12.1,7.4);
\draw[->, dashed] (P11) -- (14.5,7.4);

\draw[->, dashed] (P24) -- (E31);
\draw[->, dashed] (P24) -- (E32);
\draw[->, dashed] (P24) -- (E33);
\draw[->, dashed] (P24) -- (E34);
\draw[->, dashed] (P24) -- (8.6,3.4);
\draw[->, dashed] (P24) -- (9.0,3.4);
\draw[->, dashed] (P24) -- (9.4,3.4);
\draw[->, dashed] (P24) -- (11.7,3.4);

\end{tikzpicture}}
\caption{Recursive calls used to compute values $C_{\ell}\left(\mcali\right)$.}
\label{fig:recurrence}
\end{figure}
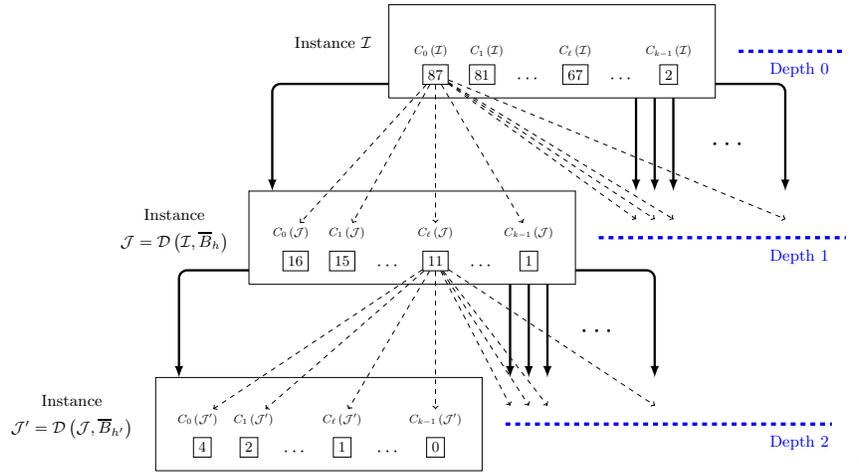

Then, we present the proof of Theorem~\ref{th:complexity} which allows us to declare the fixed-parameter tractability of \cmincuts .

\begin{theorem}
There are at most $2^{p^2}m$ instances in the recursive tree.
\label{th:complexity}
\end{theorem}
\begin{proof}[Proof of Theorem~\ref{th:complexity}]
The depth of instance $\mcali$ in the recursive tree is zero, we say $\Delta(\mcali) = 0$. 
For any closest dam $B_h$ of $\mcalz(\mcali)$, the depth of the dry instance of $B_h$ is one:  $\Delta\left(\mcald\left(\mcali, \overline{B}_{h}\right)\right) = 1$. 
More generally, if $\mcalj$ is an instance of depth $d \ge 0$ and $B_h$ a closest dam of $\mcalz(\mcalj)$, then instance $\mcald\left(\mcalj,B_h\right)$ is at depth $d + 1$.

We prove that, for any edge $e$ in graph $G$, there are at most $2^{pd}$ instances $\mcalj$ of depth $d$ such that edge $e$ belongs to the graph of $\mcalj$. This fact makes the total number of instances in the recursive tree be upper-bounded by $2^{p^2}m$.

We proceed by induction. There is one instance defined for depth $d=0$: it is $\mcali$ and it obviously contains edge $e$, so the number of instances with depth $d=0$ containing $e$ is thus $2^{pd} = 1$.
Let $d \ge 1$ and $\mcalj'$ be an instance of depth $d$ containing edge $e$: there is an instance $\mcalj$ of depth $d-1$ and one of its closest dam $B_h$ such that $\mcalj' = \mcald\left(\mcalj,\overline{B}_h\right)$. As the graph of instance $\mcalj'$ is a subgraph of those of $\mcalj$, the latter contains edge $e$.

Using the induction hypothesis, there are at most $2^{p(d-1)}$ instances $\mcalj$ of depth $d-1$ containing edge $e$. Now, given an instance $\mcalj$ with $\Delta(\mcalj) = d-1$, we bound the number of dams $\overline{B}_h$ (they are closest dams according to Lemma~\ref{le:complementary}) of $\mcalj$ such that $\mcald\left(\mcalj,\overline{B}_h\right)$ contains edge $e$. We distinguish two cases:
\begin{itemize}
\item Case 1: edge $e$ belongs to a minimum drainage cut $Z_i$ of instance $\mcalj$. We focus on the dams $B_i$ of level $i$ containing edge $e$. Their cardinality is bounded by $2^p$. The edges of level $i$ belonging to the dry instance of $\overline{B}_h$, $\mcald\left(\mcalj,\overline{B}_h\right)$, form one of these dams $B_i$. As each dam $B_i$ admits a unique closest dam (Lemma~\ref{le:unique_closest}), there cannot be more than $2^p$ closest dams $\overline{B}_h$ such that $\mcald\left(\mcalj,\overline{B}_h\right)$ contains edge $e$.
\item Case 2: edge $e$ is located between two minimum drainage cuts $Z_{i}$ and $Z_{i+1}$, $e \in R_{i+1}$. Consequently, the level of any closest dam $\overline{B}_h$ such that $\mcald\left(\mcalj,\overline{B}_h\right)$ contains $e$ is less than $i$: $i\ge h$. Therefore, the edges of level $i$ belonging to the dry instance of $\overline{B}_h$ form a dam $B_i$. Thus, the argument used in Case 1 arises the same conclusion: there cannot be more than $2^p$ closest dams such that $\mcald\left(\mcalj,\overline{B}_h\right)$ contains edge $e$.
\end{itemize}
Finally, the number of instances written as $\mcalj' = \mcald\left(\mcalj,\overline{B}_h\right)$ where $\Delta(\mcalj)=d-1$ and $\mcalj'$ contains $e$, is upper-bounded by $2^{p(d-1)}2^p = 2^{pd}$. We conclude that there are less than $2^{pd}$ instances of depth $d$ containing edge $e$. The total number of instances is thus smaller than $\sum_{d=0}^{p-1} 2^{pd}m \le 2^{p^2}m$.
\end{proof}

For any instance $\mcalj$ of the recursive tree, the algorithm computes its drainage $\mcalz\left(\mcalj\right)$, its Menger's paths and all instances $\mcald(\mcalj,B_h)$ where $B_h$ is a closest dam of $\mcalz\left(\mcalj\right)$. This third operation is done by enumerating all dams $B_i$ of $\mcalz(\mcalj)$, verifying whether there is another dam $B_h$ which is closer than $B_i$, and (if $B_i$ is a closest dam) identifying the vertices/edges of $\mcald(\mcalj,B_i)$ thanks to a depth-first search in $G_D\backslash B_i$. As there are at most $2^pn$ dams in $\mcalz\left(\mcalj\right)$, its execution time is $O(2^{2p}n^3)$. The overall complexity is $O\left(2^{p^2}m(mnp + 2^{2p}n^3)\right) = O\left(2^{p(p+2)}pmn^3\right)$.

\section{Sampling minimum edge $(S,T)$-cuts in undirected graphs} \label{sec:sampling}

We sketch the algorithm which produces one of the minimum $(S,T)$-cuts according to the uniform distribution over all minimum $(S,T)$-cuts. We run our counting algorithm (Section~\ref{subsec:algorithm}) and execute a post-processing, described below.

As in Section~\ref{subsec:algorithm}, we distinguish the input instance $\mcali$ from the other instances $\mcalj$ of the recursive tree. 
Our method to sample minimum cuts consists in searching in the recursive tree, already filled out with values $C_{\ell}(\mcalj)$ during the counting. A minimum cut of $\mcali$ is extracted thanks to a randomly driven descent in the recursive tree.

We start at root $C_{0}(\mcali)$. With probability $\frac{k}{C_{0}(\mcali)}$, the sampling algorithm returns one of the minimum drainage cuts of $\mcalz(\mcali)$ taken uniformly over them. Said differently, each cut $Z_i$ has probability $\frac{1}{C_{0}(\mcali)}$ to be produced. With probability $1-\frac{k}{C_{0}(\mcali)}$, we will go one step down the tree. Concretely, for any dam $B_i$ of $\mcalz(\mcali)$ and its closest dam $B_h$, we visit node $C_{i-h+1}\left(\mcald(\mcali,\overline{B}_h)\right)$ of depth 1 with probability $\frac{C_{i-h+1}\left(\mcald(\mcali,\overline{B}_h)\right)}{C_{0}(\mcali)}$. 
The sampling algorithm returns the union of $B_i$ with the cut obtained by a recursive call on $C_{i-h+1}(\mcald(\mcali,\overline{B}_h))$. The algorithm applied on $C_{i-h+1}(\mcald(\mcali,\overline{B}_h))$ either selects a minimum drainage cut of level greater than $i-h+1$ in $\mcald(\mcali,\overline{B}_h)$ (uniform selection among these cuts) or visits a node at depth 2, etc.

In this way, we ensure that the minimum $(S,T)$-cuts are sampled uniformly. Indeed, a cut with front dam $B_i$ is chosen with probability $\frac{C_{i-h+1}\left(\mcald(\mcali,\overline{B}_h)\right)}{C_{0}(\mcali)}$ which is the ratio of the number $C_{i-h+1}\left(\mcald(\mcali,\overline{B}_h)\right)$ of minimum cuts with front dam $B_i$ by the total number $C_{0}(\mcali)$ of minimum cuts in instance $\mcali$.

\section{Conclusion} \label{sec:conclusion}

In this study, we were interested in the parameterized complexity of counting the minimum $(S,T)$-cuts in undirected graphs. The conclusion is that an algorithm running in $O(2^{p(p+2)}pmn^3)$ was devised.
Our algorithm starts by ``draining'' the graph: the drainage is made of $k < n$ minimum cuts $Z_i$. For any minimum cut of the instance, at least one of the minimum drainage cuts $Z_i$ contains edges of $X$. For this reason, we believe that the drainage could be used on other cut problems. 
We already used it to sample minimum edge $(S,T)$-cuts.

Our work gives rise to questions concerning the counting of minimum edge $(S,T)$-cuts in undirected graphs. These questions are:

\begin{enumerate}
\item Is there an FPT$\langle p \rangle$ algorithm solving \cmincuts\ with smaller polynomial factors? In particular, for dense graphs, our algorithm complexity is $O(n^5)$ if we neglect the function of $p$. At first sight, it seems difficult to avoid the use of a maximum flow algorithm~\cite{FoFu56}, which is $O(mp)$. Is it possible to identify an algorithm with running time $O(f(p)n^2)$, $O(f(p)n^3)$, $O(f(p)n^4)$?
\item Is there an algorithm solving \cmincuts\ in time $2^{o(p^2)}n^{O(1)}$? In other words, we wonder whether it is possible to lower the factor $O(p^2)$ in the exponential and to find an algorithm with running time $2^{O(p)}n^{O(1)}$ or $2^{O(p\log p)}n^{O(1)}$, for example.
\end{enumerate}

Our algorithm can be generalized to undirected graphs with positive integer weights. The idea is to transform these undirected weighted graphs into undirected multigraphs: if the weight of edge $e=(u,v)$ is larger than $p+1$, then we replace $e$ by $p+1$ edges $(u,v)$ in parallel. Otherwise, if its weight is $w \le p$, we replace it by $w$ edges $(u,v)$. By this method, the number of edges in the transformed graph is upper-bounded by $m(p+1)$ and it preserves the minimum cuts of the input graph. Moreover, the techniques used  in this article to compute the number of minimum $(S,T)$-cuts, such as flow algorithms or closest cuts, work on multigraphs.

We can modify slightly our algorithm to enumerate minimum $(S,T)$-cuts. This simply consists in stacking in the recursive tree the set of cuts counted with $C_{\ell}(\mcali)$ instead of value $C_{\ell}(\mcali)$ itself. The number of instances stays FPT$\langle p \rangle$ but the sets of cuts put inside the recursive tree may contain a number $\Omega(n^p)$ of elements.

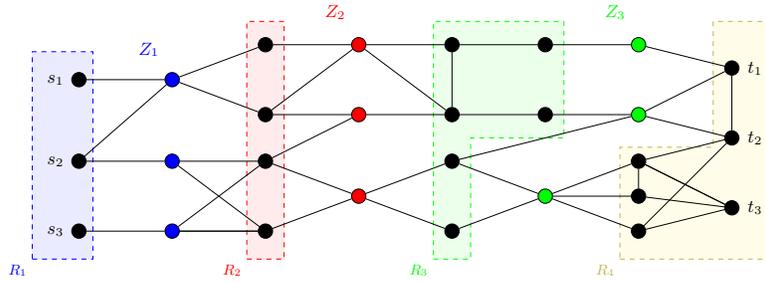
\begin{figure}[t]
\centering
\scalebox{.62}{\begin{tikzpicture}


\draw [dashed, color = blue, fill = white!92!blue] (1.0,0.4) -- (2.3,0.4) -- (2.3,4.85) -- (1.0,4.85) -- (1.0,0.4) node[below left] {$R_1$};
\draw [dashed, color = red, fill = white!92!red] (5.6,0.4) -- (6.4,0.4) -- (6.4,5.5) -- (5.6,5.5) -- (5.6,0.4) node[below left] {$R_2$};
\draw [dashed, color = green, fill = white!92!green] (9.6,0.4) -- (10.4,0.4) -- (10.4,3) -- (12.4,3) -- (12.4,5.5) -- (9.6,5.5) -- (9.6,0.4) node[below left] {$R_3$};
\draw [dashed, color = black!30!yellow, fill = white!90!yellow] (13.6,0.4) -- (16.9,0.4) -- (16.9,5.5) -- (15.6,5.5) -- (15.6,2.8) -- (13.6,2.8) -- (13.6,0.4) node[below left] {$R_4$};


\node[draw, circle, minimum height=0.2cm, minimum width=0.2cm, fill=black] (P11) at (2*1,4.25) {};
\node[draw, circle, minimum height=0.2cm, minimum width=0.2cm, fill=black] (P12) at (2*1,2.5) {};
\node[draw, circle, minimum height=0.2cm, minimum width=0.2cm, fill=black] (P13) at (2*1,1) {};

\node[draw, circle, minimum height=0.2cm, minimum width=0.2cm, fill=blue] (P21) at (2*2,4.25) {};
\node[draw, circle, minimum height=0.2cm, minimum width=0.2cm, fill=blue] (P22) at (2*2,2.5) {};
\node[draw, circle, minimum height=0.2cm, minimum width=0.2cm, fill=blue] (P23) at (2*2,1) {};

\node[draw, circle, minimum height=0.2cm, minimum width=0.2cm, fill=black] (P31) at (2*3,5) {};
\node[draw, circle, minimum height=0.2cm, minimum width=0.2cm, fill=black] (P32) at (2*3,3.5) {};
\node[draw, circle, minimum height=0.2cm, minimum width=0.2cm, fill=black] (P33) at (2*3,2.5) {};
\node[draw, circle, minimum height=0.2cm, minimum width=0.2cm, fill=black] (P34) at (2*3,1) {};

\node[draw, circle, minimum height=0.2cm, minimum width=0.2cm, fill=red] (P41) at (2*4,5) {};
\node[draw, circle, minimum height=0.2cm, minimum width=0.2cm, fill=red] (P42) at (2*4,3.5) {};
\node[draw, circle, minimum height=0.2cm, minimum width=0.2cm, fill=red] (P43) at (2*4,1.75) {};

\node[draw, circle, minimum height=0.2cm, minimum width=0.2cm, fill=black] (P51) at (2*5,5) {};
\node[draw, circle, minimum height=0.2cm, minimum width=0.2cm, fill=black] (P52) at (2*5,3.5) {};
\node[draw, circle, minimum height=0.2cm, minimum width=0.2cm, fill=black] (P53) at (2*5,2.5) {};
\node[draw, circle, minimum height=0.2cm, minimum width=0.2cm, fill=black] (P54) at (2*5,1) {};

\node[draw, circle, minimum height=0.2cm, minimum width=0.2cm, fill=black] (P61) at (2*6,5) {};
\node[draw, circle, minimum height=0.2cm, minimum width=0.2cm, fill=black] (P62) at (2*6,3.5) {};
\node[draw, circle, minimum height=0.2cm, minimum width=0.2cm, fill=green] (P63) at (2*6,1.75) {};

\node[draw, circle, minimum height=0.2cm, minimum width=0.2cm, fill=green] (P71) at (2*7,5) {};
\node[draw, circle, minimum height=0.2cm, minimum width=0.2cm, fill=green] (P72) at (2*7,3.5) {};
\node[draw, circle, minimum height=0.2cm, minimum width=0.2cm, fill=black] (P73) at (2*7,2.5) {};
\node[draw, circle, minimum height=0.2cm, minimum width=0.2cm, fill=black] (P74) at (2*7,1.75) {};
\node[draw, circle, minimum height=0.2cm, minimum width=0.2cm, fill=black] (P75) at (2*7,1) {};

\node[draw, circle, minimum height=0.2cm, minimum width=0.2cm, fill=black] (P81) at (2*8,4.5) {};
\node[draw, circle, minimum height=0.2cm, minimum width=0.2cm, fill=black] (P82) at (2*8,3.0) {};
\node[draw, circle, minimum height=0.2cm, minimum width=0.2cm, fill=black] (P83) at (2*8,1.5) {};


\draw (P11) -- (P21);
\draw (P21) -- (P31);
\draw (P31) -- (P41);
\draw (P41) -- (P51);
\draw (P51) -- (P61);
\draw (P61) -- (P71);
\draw (P71) -- (P81);

\draw (P12) -- (P21);
\draw (P21) -- (P32);
\draw (P32) -- (P42);
\draw (P33) -- (P42);
\draw (P42) -- (P52);
\draw (P52) -- (P62);
\draw (P62) -- (P72);
\draw (P72) -- (P81);
\draw (P72) -- (P82);

\draw (P12) -- (P22);

\draw (P22) -- (P33);
\draw (P33) -- (P43);
\draw (P43) -- (P53);
\draw (P43) -- (P54);
\draw (P53) -- (P63);
\draw (P63) -- (P74);
\draw (P74) -- (P73);
\draw (P54) -- (P63);
\draw (P63) -- (P75);
\draw (P75) -- (P83);
\draw (P75) -- (P82);

\draw (P13) -- (P23);
\draw (P23) -- (P34);
\draw (P23) -- (P33);
\draw (P34) -- (P43);
\draw (P63) -- (P73);
\draw (P73) -- (P82);
\draw (P73) -- (P83);

\draw (P22) -- (P34);
\draw (P32) -- (P41);
\draw (P41) -- (P52);
\draw (P52) -- (P51);
\draw (P53) -- (P72);
\draw (P23) -- (P34);

\draw (P74) -- (P83);
\draw (P73) -- (P83);
\draw (P81) -- (P82);


\node[scale=1.2] at (1.5,4.25) {$s_1$};
\node[scale=1.2] at (1.5,2.5) {$s_2$};
\node[scale=1.2] at (1.5,1) {$s_3$};

\node[scale=1.2] at (16.5,4.5) {$t_1$};
\node[scale=1.2] at (16.5,3.0) {$t_2$};
\node[scale=1.2] at (16.5,1.5) {$t_3$};

\node[scale=1.2, color = blue] at (3.5,4.9) {$Z_1$};
\node[scale=1.2, color = red] at (7.5,5.7) {$Z_2$};
\node[scale=1.2, color = green] at (13.5,5.7) {$Z_3$};

\end{tikzpicture}}
\caption{An example of drainage when cuts are composed of vertices.}
\label{fig:skeleton_vertex}
\end{figure}

We believe that the techniques proposed in this study could be used to count minimum vertex $(S,T)$-cuts in undirected graphs and minimum $(S,T)$-cuts in directed graphs. However, we explain below why major changes of our algorithm are needed to make it work on these applications.

Let us focus on minimum vertex $(S,T)$-cuts. The drainage can be extended to vertex cuts, as the unicity of the minimum closest cut is preserved. Fig.~\ref{fig:skeleton_vertex} illustrates how the drainage could be defined for minimum vertex $(S,T)$-cuts: cut $Z_1$ is the minimum closest $(S,T)$-cut, set $R_1$ is equal to $R(Z_1,S_1)$ with $S_1=S$. Then, $S_2 = Z_1$, cut $Z_2$ is the minimum closest $(S_2,T)$-cut in graph $G\backslash R(Z_1,S)$, and so on.

With this definition, the drainage fulfils the properties given in Section~\ref{sec:skeleton} for the edge version, the most important of them being that any minimum vertex $(S,T)$-cut $X$ admits a front dam $B_{i(X)}$. Moreover, there is a vertex version of the Menger's theorem, stating that the size of the minimum $(S,T)$-cut is equal to the maximum number of vertex-disjoint $(S,T)$-paths. The largest set of maximum vertex-disjoint $(S,T)$-paths is computed in polynomial time and, consequently, the definitions of the dry area and the dry instance can be naturally extended.

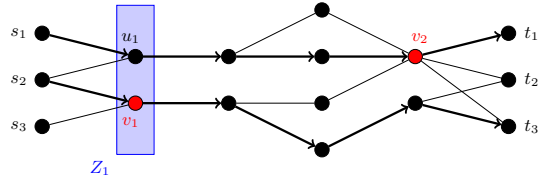
\begin{figure}[h]
\centering
\scalebox{.62}{\begin{tikzpicture}


\draw [color = blue, fill = white!80!blue] (3.6,0.4) -- (4.4,0.4) -- (4.4,3.6) -- (3.6,3.6) -- (3.6,0.4) node[below left,scale=1.2] {$Z_1$};


\node[draw, circle, minimum height=0.2cm, minimum width=0.2cm, fill=black] (P11) at (2*1,3) {};
\node[draw, circle, minimum height=0.2cm, minimum width=0.2cm, fill=black] (P12) at (2*1,2) {};
\node[draw, circle, minimum height=0.2cm, minimum width=0.2cm, fill=black] (P13) at (2*1,1) {};

\node[draw, circle, minimum height=0.2cm, minimum width=0.2cm, fill=black] (P21) at (2*2,2.5) {};
\node[draw, circle, minimum height=0.2cm, minimum width=0.2cm, fill=red] (P22) at (2*2,1.5) {};

\node[draw, circle, minimum height=0.2cm, minimum width=0.2cm, fill=black] (P31) at (2*3,2.5) {};
\node[draw, circle, minimum height=0.2cm, minimum width=0.2cm, fill=black] (P32) at (2*3,1.5) {};

\node[draw, circle, minimum height=0.2cm, minimum width=0.2cm, fill=black] (P41) at (2*4,3.5) {};
\node[draw, circle, minimum height=0.2cm, minimum width=0.2cm, fill=black] (P42) at (2*4,2.5) {};
\node[draw, circle, minimum height=0.2cm, minimum width=0.2cm, fill=black] (P43) at (2*4,1.5) {};
\node[draw, circle, minimum height=0.2cm, minimum width=0.2cm, fill=black] (P44) at (2*4,0.5) {};

\node[draw, circle, minimum height=0.2cm, minimum width=0.2cm, fill=red] (P51) at (2*5,2.5) {};
\node[draw, circle, minimum height=0.2cm, minimum width=0.2cm, fill=black] (P52) at (2*5,1.5) {};

\node[draw, circle, minimum height=0.2cm, minimum width=0.2cm, fill=black] (P61) at (2*6,3.0) {};
\node[draw, circle, minimum height=0.2cm, minimum width=0.2cm, fill=black] (P62) at (2*6,2.0) {};
\node[draw, circle, minimum height=0.2cm, minimum width=0.2cm, fill=black] (P63) at (2*6,1.0) {};


\draw[->,line width = 1.4pt] (P11) -- (P21);
\draw (P12) -- (P21);
\draw[->,line width = 1.4pt] (P12) -- (P22);
\draw (P13) -- (P22);

\draw[->,line width = 1.4pt] (P21) -- (P31);
\draw[->,line width = 1.4pt] (P22) -- (P32);

\draw (P31) -- (P41);
\draw[->,line width = 1.4pt] (P31) -- (P42);
\draw (P32) -- (P43);
\draw[->,line width = 1.4pt] (P32) -- (P44);

\draw (P41) -- (P51);
\draw[->,line width = 1.4pt] (P42) -- (P51);
\draw (P43) -- (P51);
\draw[->,line width = 1.4pt] (P44) -- (P52);

\draw[->,line width = 1.4pt] (P51) -- (P61);
\draw (P51) -- (P62);
\draw (P51) -- (P63);
\draw (P52) -- (P62);
\draw[->,line width = 1.4pt] (P52) -- (P63);


\node[scale=1.2] at (1.5,3.0) {$s_1$};
\node[scale=1.2] at (1.5,2.0) {$s_2$};
\node[scale=1.2] at (1.5,1.0) {$s_3$};

\node[scale=1.2] at (12.5,3.0) {$t_1$};
\node[scale=1.2] at (12.5,2.0) {$t_2$};
\node[scale=1.2] at (12.5,1.0) {$t_3$};

\node[scale=1.2, color = red] at (3.9,1.1) {$v_1$};
\node[scale=1.2, color = red] at (10.1,2.9) {$v_2$};
\node[scale=1.2, color = black] at (3.9,2.9) {$u_1$};

\end{tikzpicture}}
\caption{Illustration of the impossibility to prove Theorem~\ref{th:keystone} for vertex cuts}
\label{fig:keystone_vertex}
\end{figure}

However, Theorem~\ref{th:keystone} does not hold anymore: the set $X\backslash B_{i(X)}$ of a minimum $(S,T)$-cut $X$ is not necessarily included in the dry instance of $\overline{B}_{h(X)}$. We give an example in Fig.~\ref{fig:keystone_vertex}. Set $X = \set{v_1,v_2}$ is a minimum vertex $(S,T)$-cut and its two vertices are drawn in red. One of its vertices $v_1$ belongs to cut $Z_1$, so $B_{i(X)} = B_{h(X)} = \set{v_1}$. We observe that vertex $v_2$ is reachable from $S$ in graph $G_D\backslash \overline{B}_{h(X)}$ because the dry area of $\overline{B}_{h(X)} = \set{u_1}$ does not contain any vertex different than $u_1$ itself. To pursue the work presented in this article, our intention is to modify the definitions of the dry area and the dry instance in order to make Theorem~\ref{th:keystone} be true for vertex cuts.

\section*{Acknowledgement} We would like to thank Saket Saurabh who pointed out the related work~\cite{MaORa13} during the WG'19 conference.

\bibliographystyle{splncs04}
\bibliography{countingedgecuts}

\end{document}